\documentclass[11pt]{article}
\usepackage[latin1]{inputenc}   % Pour pouvoir utiliser directement les accents
\usepackage{amsmath}   % Pour les maths
\usepackage{amsfonts}  % Symboles
\usepackage{amssymb}  % Symboles
\usepackage{latexsym}  % Symboles
\usepackage{graphicx} % Pour inclure des figures eps
\usepackage{epstopdf}
\usepackage[a4paper]{geometry}
\usepackage{color} 
\usepackage{subfigure}
\usepackage{url}
\usepackage{minitoc}
\usepackage{longtable}

\usepackage{hyperref} %% pour PDF

\newtheorem{theorem}{Theorem}[section]
\newtheorem{proposition}[theorem]{Proposition}
\newtheorem{conjecture}[theorem]{Conjecture}
\newtheorem{lemma}[theorem]{Lemma}

\newtheorem{claim}[theorem]{Claim}
\newtheorem{definition}[theorem]{Definition}

\newcommand{\mc}[1]{\mathcal{#1}}
\newcommand{\mb}[1]{\mathbb{#1}}
\newcommand{\pdc}[1]{\hat{#1}}
\newcommand{\Out}{\normalfont{Out}}
\newcommand{\ccw}{counterclockwise }

\newcommand{\cw}{clockwise }
\newcommand{\cww}{clockwise}

\newcounter{sclaim}
\newcounter{ssclaim}

\newenvironment{proof}{\noindent \setcounter{ssclaim}{0}\emph{Proof.}\ }{\hfill
    $\Box$\vspace{1em}}

  \newenvironment{proofclaim}{\noindent \emph{Proof.}\ }{\hfill
    $\Diamond$\vspace{1em}}

\newenvironment{sclaim}[1][]%
{\refstepcounter{sclaim}\vspace{1ex}\noindent{\it  (\arabic{sclaim})  {#1}{}}\it}{\vspace{1ex}}

\newenvironment{proofsclaim}[1][]%
	{\noindent {}{#1}{}}{ This proves~(\arabic{sclaim}).\vspace{1ex}}

\title{On the structure of Schnyder woods on orientable surfaces\thanks{This
      work was supported by the grant EGOS ANR-12-JS02-002-01}}

  \author{Daniel Gon\c{c}alves\thanks{CNRS, Université de Montpellier,
      LIRMM UMR 5506, CC477, 161 rue Ada, 34095 Montpellier Cedex 5,
      France.  \texttt{daniel.goncalves@lirmm.fr}}, Kolja
    Knauer\thanks{Aix Marseille Universit\'e, LIF UMR 7279, Parc
      Scientifique et Technologique de Luminy, 163 avenue de Luminy -
      Case 901, 13288 Marseille Cedex 9, France.
      \texttt{kolja.knauer@lif.univ-mrs.fr}}, Benjamin
    L\'ev\^eque\thanks{CNRS, Laboratoire G-SCOP UMR 5272, 46 Avenue Félix
      Viallet, 38031 Grenoble Cedex 1, France
      \texttt{benjamin.leveque@cnrs.fr}}}

\begin{document}
\maketitle

\begin{abstract}
  We propose a simple generalization of Schnyder woods from the plane
  to maps on orientable surfaces of higher genus. This is done in the
  language of angle labelings.  Generalizing results of De Fraysseix
  and Ossona de Mendez, and Felsner, we establish a correspondence
  between these labelings and orientations and characterize the set of
  orientations of a map that correspond to such a Schnyder
  labeling. Furthermore, we study the set of these orientations of a
  given map and provide a natural partition into distributive lattices
  depending on the surface homology.  This generalizes earlier results
  of Felsner and Ossona de Mendez.  In the toroidal case, a new proof
  for the existence of Schnyder woods is derived from this approach.
\end{abstract}

\section{Introduction}

Schnyder~\cite{Sch89} introduced Schnyder woods for planar
triangulations with
the following local property:

\begin{definition}[Schnyder property]
\label{def:schnyderproperty}
Given a map $G$, a vertex $v$ and an orientation and
coloring\footnote{Throughout the paper colors and some of the indices
  are given modulo 3.} of the edges incident to $v$ with the colors
$0$, $1$, $2$, we say that $v$ satisfies the \emph{Schnyder property},
(see Figure~\ref{fig:LSP}) if $v$ satisfies the following local
property:

\begin{itemize}
\item Vertex $v$ has out-degree one in each color.
\item The edges $e_0(v)$, $e_1(v)$, $e_2(v)$ leaving $v$ in colors
  $0$, $1$, $2$, respectively, occur in counterclockwise order.
\item Each edge entering $v$ in color $i$ enters $v$ in the
  counterclockwise sector from $e_{i+1}(v)$ to
  $e_{i-1}(v)$.
\end{itemize}
\end{definition}

\begin{figure}[!h]
\center
\includegraphics[scale=0.5]{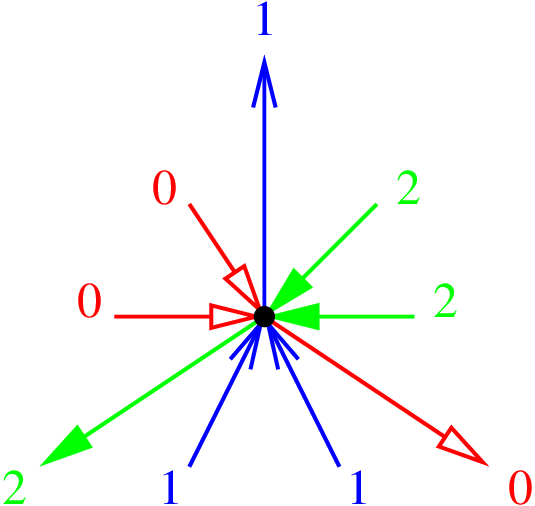}
\caption{The Schnyder property. The depicted correspondence between red, blue,
  green, 0, 1, 2, and the arrow shapes will be used through the paper.}
\label{fig:LSP}
\end{figure}

\begin{definition}[Schnyder wood]
\label{def:schnyder}
Given a planar triangulation $G$, a \emph{Schnyder wood} is an
orientation and coloring of the inner edges of $G$ with the colors
$0$, $1$, $2$ (edges are oriented in one direction only), where each
inner vertex $v$ satisfies the \emph{Schnyder property}.
\end{definition}

See Figure~\ref{fig:planar-triangulation} for an example of a Schnyder
wood.

\begin{figure}[!h]
\center
\includegraphics[scale=0.6]{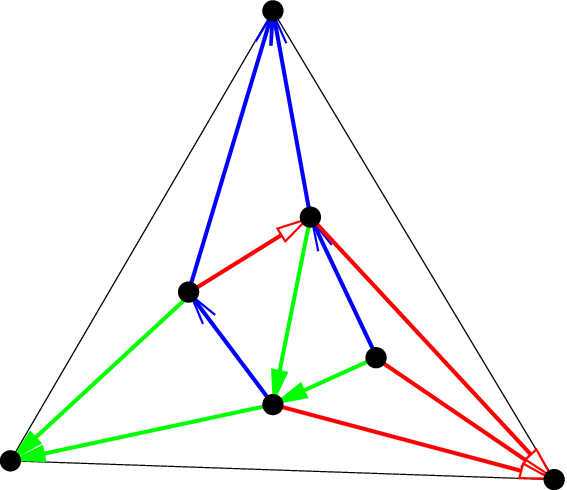}
\caption{Example of a Schnyder wood of a planar triangulation.}
\label{fig:planar-triangulation}
\end{figure}

Schnyder woods are today one of the main tools in the area of planar
graph representations. Among their most prominent applications are the
following: They provide a machinery to construct space-efficient
straight-line drawings~\cite{Sch90, Kan96,Fel01}, yield a
characterization of planar graphs via the dimension of their
vertex-edge incidence poset~\cite{Sch89,Fel01}, and are used to encode
triangulations~\cite{PS06,Bar12}. Further applications lie in
enumeration~\cite{Bon05}, representation by geometric
objects~\cite{FOR94, GLP11}, graph spanners~\cite{BGHI10}, etc.  The
richness of these applications has stimulated research towards
generalizing Schnyder woods to non planar graphs.

For higher genus triangulated surfaces, a generalization of Schnyder
woods has been proposed by Castelli Aleardi, Fusy and
Lewiner~\cite{CFL09}, with applications to encoding.  In this
definition, the simplicity and the symmetry of the original definition
of Schnyder woods are lost. Here we propose an alternative
generalization of Schnyder woods for higher genus that generalizes the
one proposed in~\cite{GL13} for the toroidal case.

A closed curve on a surface is \emph{contractible} if it can be
continuously transformed into a single point. Except if stated
otherwise, we consider  graphs embedded on orientable
surfaces such that they do not have contractible cycles of size $1$ or
$2$ (i.e. no contractible loops and no contractible double
edges). Note that this is a weaker assumption, than the graph being
\emph{simple}, i.e. not having \emph{any} cycles of size $1$ or $2$
(i.e. no loops and no multiple edges). A graph embedded on a surface
is called a \emph{map} on this surface if all its faces are
homeomorphic to open disks.  A map is a triangulation if all its faces
are triangles.

In this paper we consider finite maps. We denote by $n$ be the number
of vertices and $m$ the number of edges of a graph. Given a graph
embedded on a surface, we use $f$ for the number of faces.  Euler's
formula says that any map on an orientable surface of genus $g$
satisfies $n-m+f=2-2g$. In particular, the plane is the surface of
genus $0$, the torus the surface of genus $1$, the double torus the
surface of genus $2$, etc.  By Euler's formula, a triangulation of
genus $g$ has exactly $3n+6(g-1)$ edges. So having a generalization of
Schnyder woods in mind, for all $g\ge 2$ there are too many edges to
force all vertices to have outdegree exactly three. This problem can be
overcome by allowing vertices to fulfill the Schnyder property
``several times'', i.e. such vertices have outdegree 6, 9, etc. with
the color property of Figure~\ref{fig:LSP} repeated several times (see
Figure~\ref{fig:369}).

\begin{figure}[!h]
\center
\begin{tabular}{cc}
\includegraphics[scale=0.5]{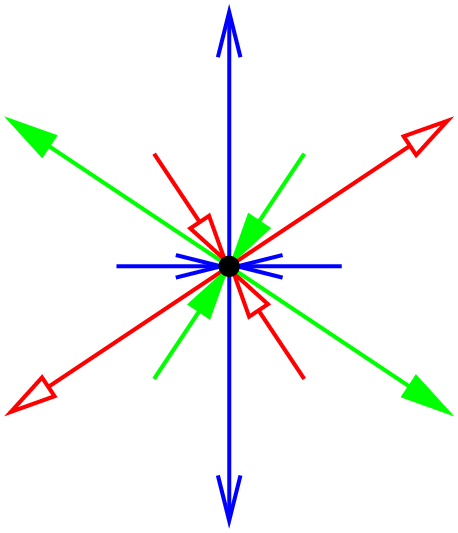} \ \ \ & \ \ \
\includegraphics[scale=0.5]{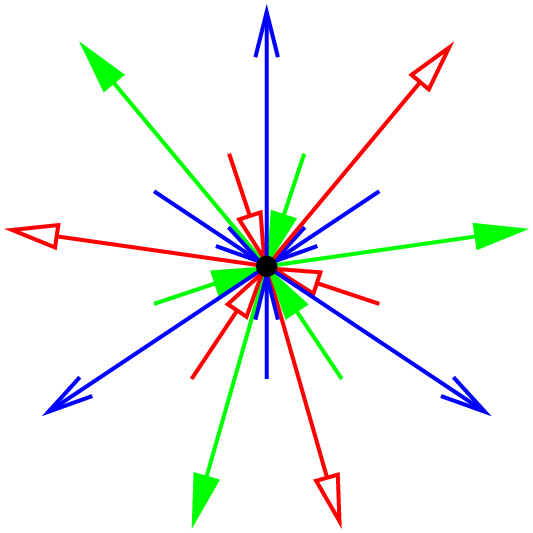} \\
Outdegree six \ \ \ & \ \ \ Outdegree nine \\
\end{tabular}
\caption{The Schnyder property repeated several times around a
  vertex.}
\label{fig:369}
\end{figure}

Figure~\ref{fig:doubletorus} is an example of such a Schnyder wood on
a triangulation of the double torus. The double torus is represented
by a fundamental polygon -- an octagon. The sides of the octagon are
identified according to their labels. All the vertices of the
triangulation have outdegree three except two vertices, the circled ones,
that have outdegree six. Each of the latter appear twice in the
representation.

\begin{figure}[!h]
\center
\includegraphics[scale=0.3]{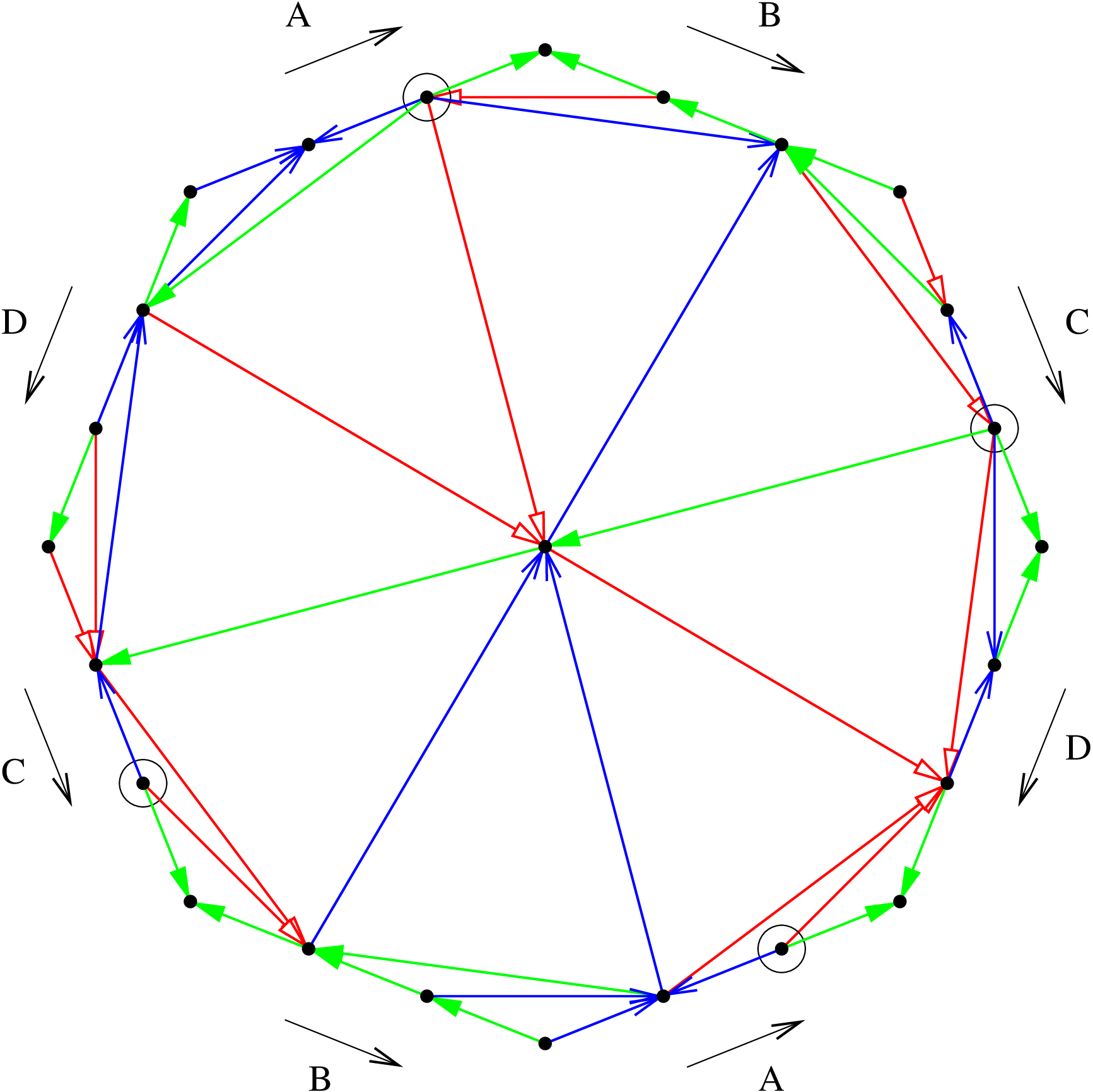}
\caption{A Schnyder wood of a triangulation of the double torus.}
\label{fig:doubletorus}
\end{figure}

In this paper we formalize this idea to obtain a concept of Schnyder
woods applicable to general maps (not only triangulations) on
arbitrary orientable surfaces. This is based on the definition of
Schnyder woods via angle labelings in
Section~\ref{sec:generalization}. We prove several basic properties of
these objects. While every map admits a ``trivial'' Schnyder wood, the
existence of a non-trivial one remains open but
leads to interesting conjectures.

By a result of De Fraysseix and Ossona de Mendez~\cite{FO01}, for any
planar triangulation there is a bijection between its Schnyder woods
and the orientations of its inner edges where every inner vertex has
outdegree three. Thus, any orientation with the proper outdegree
corresponds to a Schnyder wood and there is a unique way, up to
symmetry of the colors, to assign colors to the oriented edges in
order to fulfill the Schnyder property at every inner vertex. This is
not true in higher genus as already in the torus, there exist
orientations that do not correspond to any Schnyder wood (see
Figure~\ref{fig:orientation}).  In Section~\ref{sec:char}, we
characterize orientations that correspond to our generalization of
Schnyder woods.

\begin{figure}[!h]
\center
\includegraphics[scale=0.5]{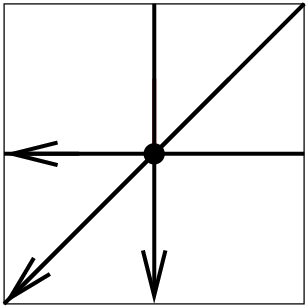} 
\ \ \ \ 
\includegraphics[scale=0.5]{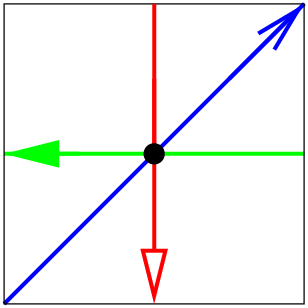}
\caption{Two different orientations of a toroidal triangulation. Only
  the one on the right corresponds to a Schnyder wood.}
\label{fig:orientation}
\end{figure}

In Section~\ref{sec:stru}, we study the transformations between
Schnyder orientations. We obtain a partition of the set of Schnyder
woods into homology classes of orientations, each of these classes
being a distributive lattice.
This generalizes corresponding results obtained for the
plane by Ossona de Mendez~\cite{Oss94} and
Felsner~\cite{Fel04}. The particular properties of the minimal element of such a
lattice recently led to an optimal linear encoding method for toroidal triangulations by Despré, the first author, and the third author~\cite{DGL15}. 
This generalizes previous results of Poulalhon and Schaeffer for the plane~\cite{PS06}.

In Section~\ref{sec:toretri}, we focus on toroidal
triangulations. We use the characterization theorem of
Section~\ref{sec:char} to give a new proof of the existence of
Schnyder woods in this case. We show that the so-called ``crossing''
property allows to define a canonical lattice. Note that this special
lattice is the one used in~\cite{DGL15} to obtain a bijection. Finally
the results of the paper are illustrated by an example.

\section{Generalization of Schnyder woods}
\label{sec:generalization}

\subsection{Angle labelings}
\label{sec:anglelabeling}

Consider a map $G$ on an orientable surface. An \emph{angle labeling}
of $G$ is a labeling of the angles of $G$ (i.e. face corners of $G$)
in colors $0$, $1$, $2$.  More formally, we denote an angle labeling
by a function $\ell:\mc{A}\to \mb{Z}_3$, where $\mc{A}$ is the set of
angles of $G$.  Given an angle labeling, we define several properties
of vertices, faces and edges that generalize the notion of Schnyder
angle labeling in the planar case~\cite{Fel-book}.

Consider an angle labeling $\ell$ of $G$.
A vertex or a face $v$ is of \emph{type $k$}, for $k\geq 1$, if the
labels of the angles around $v$ form, in counterclockwise order, $3k$
nonempty intervals such that in the $j$-th interval all the angles
have color $(j \mod 3)$. A vertex or a face $v$ is of \emph{type
  $0$}, if the labels of the angles around $v$ are all of color $i$
for some $i$ in $\{0,1,2\}$. 

An edge $e$ is of \emph{type $1$ or $2$} if the labels of the four
angles incident to edge $e$ are, in clockwise order, $i-1$, $i$,
$i$, $i+1$ for some $i$ in $\{0,1,2\}$. The edge $e$ is of \emph{type
  $1$} if the two angles with the same color are incident to the same
extremity of $e$ and of \emph{type $2$} if the two angles are incident
to the same side of $e$. An edge $e$ is of
\emph{type $0$} if the labels of the four angles incident to edge $e$
are all $i$ for some $i$ in $\{0,1,2\}$ (See
Figure~\ref{fig:edgelabeling}).

If there exists a function $f:V\to \mathbb N$ such that every vertex
$v$ of $G$ is of type $f(v)$, we say that $\ell$ is $f$-{\sc
  vertex}. If we do not want to specify the function $f$, we simply say that $\ell$ is {\sc vertex}.
  We sometimes use the notation $K$-{\sc
  vertex} if the labeling is $f$-{\sc vertex} for a function $f$ with
$f(V)\subseteq K$. When $K=\{k\}$, i.e. $f$ is a constant function,
then we use the notation $k$-{\sc vertex} instead of $f$-{\sc
  vertex}. Similarly we define {\sc face}, $K$-{\sc face}, $k$-{\sc
  face}, {\sc edge}, $K$-{\sc edge}, $k$-{\sc edge}.

The following lemma expresses that property {\sc edge} is the central
notion here.  Properties $K$-{\sc vertex} and $K$-{\sc face} are used
later on to express additional requirements on the angle labelings
that are considered.

\begin{lemma}
\label{lem:EDGElabeling}
An {\sc edge} angle labeling is {\sc vertex} and {\sc face}.
\end{lemma}

\begin{proof}
  Let $\ell$ be an {\sc edge} angle labeling. Consider two
  counterclockwise consecutive angles $a,a'$ around a vertex (or a
  face). Property {\sc edge} implies that $\ell(a')=\ell(a)$ or
  $\ell(a')=\ell(a)+1$ (see Figure~\ref{fig:edgelabeling}). Thus by
  considering all the angles around a vertex or a face, it is clear
  that $\ell$ is also {\sc vertex} and {\sc face}.
\end{proof}

Thus we define a Schnyder labeling as follows:

\begin{definition}[Schnyder labeling]
\label{def:schnyderlabeling}
Given a map $G$ on an orientable surface, a \emph{Schnyder labeling} of $G$ is
an {\sc edge} angle labeling of $G$.
\end{definition}

Figure~\ref{fig:edgelabeling} shows how an {\sc edge} angle labeling
defines an orientation and coloring of the edges of the graph with
edges oriented in one direction or in two opposite directions.  

\begin{figure}[!h]
\center
\begin{tabular}{ccc}
\includegraphics[scale=0.5]{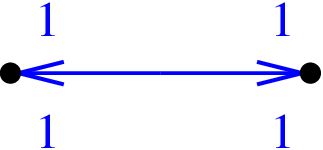} \ \ \  &  \ \ \
\includegraphics[scale=0.5]{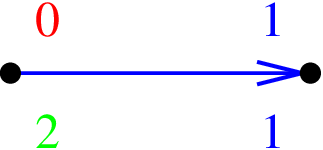} \ \ \  &  \ \ \
\includegraphics[scale=0.5]{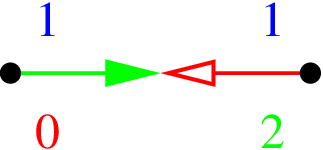} \\
Type 0  \ \ \  &  \ \ \ Type 1  \ \ \  &  \ \ \ Type 2  \\
\end{tabular}
\caption{Correspondence between {\sc edge} angle labelings and some
  bi-orientations and colorings of the edges.}
\label{fig:edgelabeling}
\end{figure}

In the next two sections, the correspondence from Figure~\ref{fig:edgelabeling}
is used to show that Schnyder labelings correspond to or generalize
previously defined Schnyder woods in the plane and in the torus.
Hence, they are a natural generalization of Schnyder woods for higher genus.

\subsection{Planar Schnyder woods}
\label{sec:plane}

Originally, Schnyder woods were defined only for planar
triangulations~\cite{Sch89}. Felsner~\cite{Fel01, Fel03} extended this
definition to planar maps. To do so he allowed edges to be
oriented in one direction or in two opposite directions (originally
only one direction was possible). The formal definition is the
following:

\begin{definition}[Planar Schnyder wood]
\label{def:felsner}
Given a planar map $G$. Let $x_0$, $x_1$, $x_2$ be three vertices
occurring in counterclockwise order on the outer face of $G$. The
\emph{suspension} $G^\sigma$ is obtained by attaching a half-edge that
reaches into the outer face to each of these special vertices.  A
\emph{planar Schnyder wood} rooted at $x_0$, $x_1$, $x_2$ is an
orientation and coloring of the edges of $G^\sigma$ with the colors
$0$, $1$, $2$, where every edge $e$ is oriented in one direction or in
two opposite directions (each direction having a distinct color and
being outgoing), satisfying the following conditions:

\begin{itemize}
\item Every vertex satisfies the Schnyder property and the
  half-edge at $x_i$ is directed outward and colored $i$.
\item There is no interior face whose boundary is a
monochromatic cycle.
\end{itemize}
\end{definition}

See Figure~\ref{fig:planar-felsner} for two examples of
planar Schnyder woods.

\begin{figure}[!h]
\center
\includegraphics[scale=0.5]{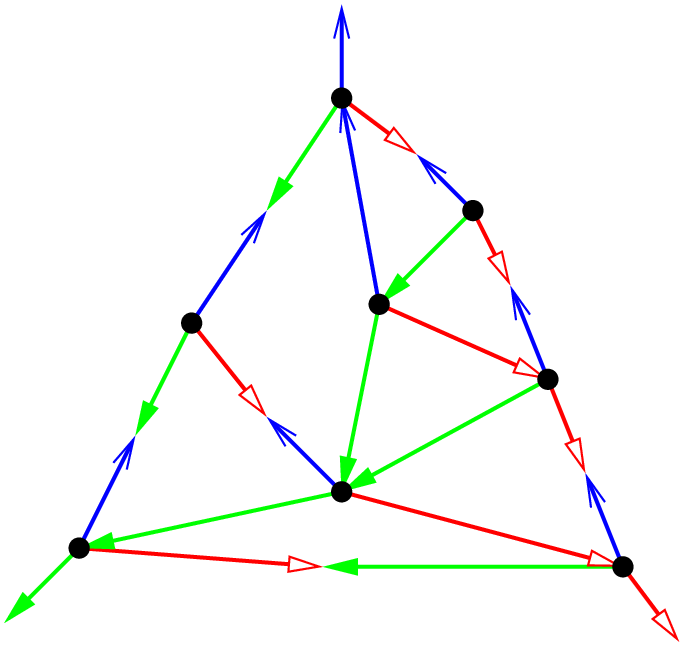}
\hspace{1cm}
\includegraphics[scale=0.5]{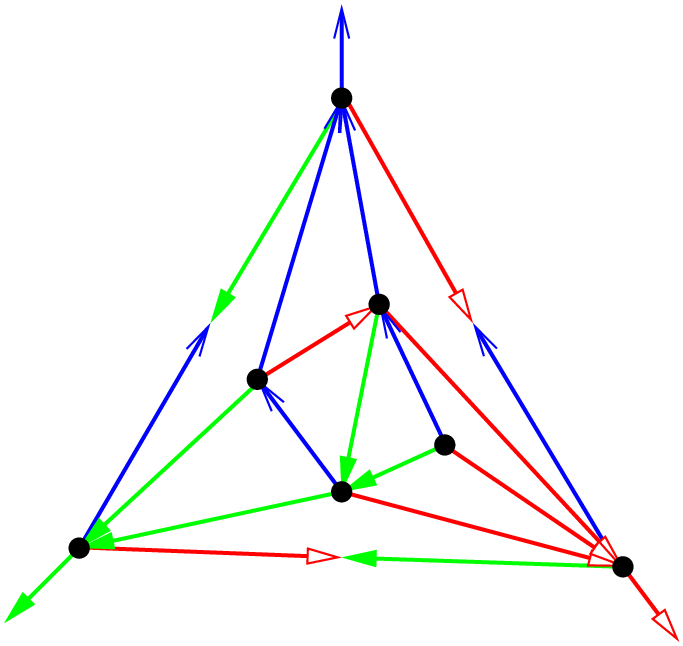}
\caption{A planar Schnyder wood of a planar map and of a planar
  triangulation.}
\label{fig:planar-felsner}
\end{figure}

The correspondence of Figure~\ref{fig:edgelabeling} gives the
following bijection, as proved by Felsner~\cite{Fel03}:

\begin{proposition}[\cite{Fel03}]
\label{prop:bijfelsner}
  If $G$ is a planar map and $x_0$, $x_1$, $x_2$ are three vertices
  occurring in counterclockwise order on the outer face of $G$, then
  the planar Schnyder woods of $G^\sigma$ are in bijection with the
  \{1,2\}-{\sc edge}, 1-{\sc vertex}, 1-{\sc face} angle labelings of
  $G^\sigma$ (with the outer face being 1-{\sc face} but in clockwise
  order).
\end{proposition}

Felsner~\cite{Fel01} and Miller~\cite{Mil02} characterized the planar
maps that admit a planar Schnyder wood. Namely, they are the
internally 3-connected maps (i.e. those with three vertices on the
outer face such that the graph obtained from $G$ by adding a vertex
adjacent to the three vertices is 3-connected).

\subsection{Generalized Schnyder woods}
\label{sec:generalSW}

Any map (on any orientable surface) admits a trivial {\sc edge} angle
labeling: the one with all angles labeled $i$ (and thus all edges,
vertices, and faces are of type 0).  A natural non-trivial case, that is also
symmetric for the duality, is to consider {\sc edge},
$\mathbb{N}^*$-{\sc vertex}, $\mathbb{N}^*$-{\sc face} angle labelings
of general maps (where $\mathbb{N}^*=\mathbb{N}\setminus\{0\}$).  In
planar Schnyder woods only type 1 and type 2 edges are used. Here we
allow type 0 edges because they seem unavoidable for some maps (see
discussion below).  This suggests the following definition of Schnyder
woods in higher genus.

First, the generalization of the Schnyder property is the following:

\begin{definition}[Generalized Schnyder property]
\label{def:generalschnyderproperty}
Given a map $G$ on a genus $g\geq 1$ orientable surface, a vertex $v$
and an orientation and coloring of the edges incident to $v$ with the
colors $0$, $1$, $2$, we say that  $v$ satisfies the
\emph{generalized Schnyder property} (see Figure~\ref{fig:369}), if
$v$ satisfies the following local property for $k\geq 1$:

\begin{itemize}
\item Vertex $v$ has out-degree $3k$.
\item The edges $e_0(v),\ldots, e_{3k-1}(v)$ leaving $v$ in
  counterclockwise order are such that $e_j(v)$ has color
  $j \bmod 3$.
\item Each edge entering $v$ in color $i$ enters $v$ in a
  counterclockwise sector from $e_{j}(v)$ to $e_{j+1}(v)$ 
  with
  $i\not\equiv j\ (\bmod 3)$ and $i \not\equiv j+1 \ (\bmod 3)$.
\end{itemize}
\end{definition}

Then, the generalization of Schnyder woods is the
following (where the three types of edges depicted on
Figure~\ref{fig:edgelabeling} are allowed):

\begin{definition}[Generalized Schnyder wood]
\label{def:highgenus}
  Given a map $G$ on a genus $g\geq 1$ orientable surface, a
  \emph{generalized Schnyder wood} of $G$ is an orientation and
  coloring of the edges of $G$ with the colors $0$, $1$, $2$, where
  every edge is oriented in one direction or in two opposite
  directions (each direction having a distinct color and being
  outgoing, or each direction having the same color and being
  incoming), satisfying the following conditions:
 
\begin{itemize}
\item Every vertex satisfies the generalized Schnyder property.
\item There is no face whose boundary is a
monochromatic cycle.
\end{itemize}
\end{definition}

When there is no ambiguity we call ``generalized Schnyder woods'' just ``Schnyder woods''.
See Figure~\ref{fig:tore-primal} for two examples of Schnyder woods in
the torus.

\begin{figure}[h!]
\center
\includegraphics[scale=0.4]{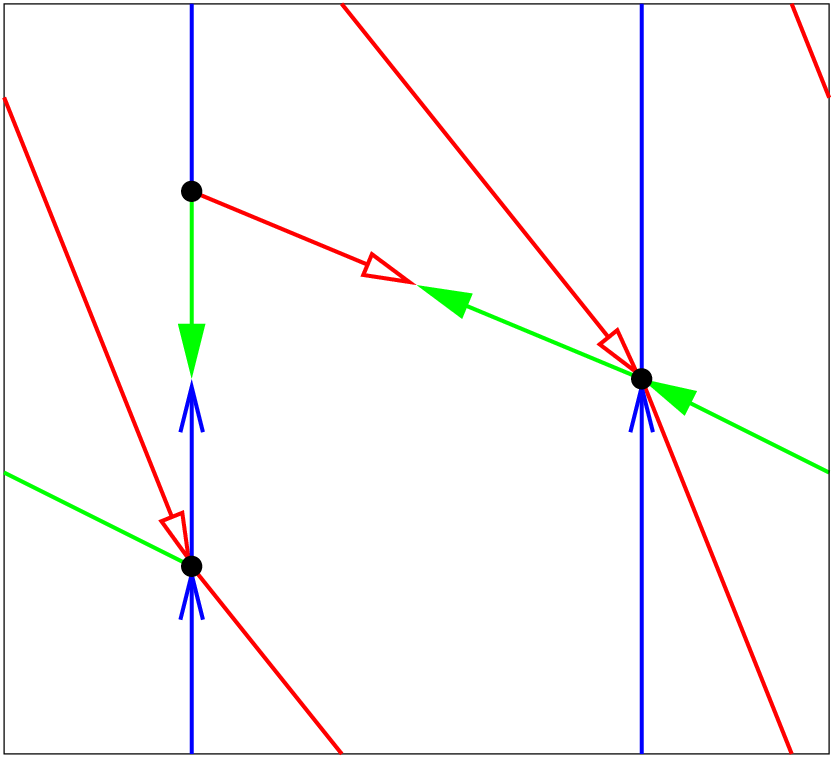}
\hspace{1cm}
\includegraphics[scale=0.4]{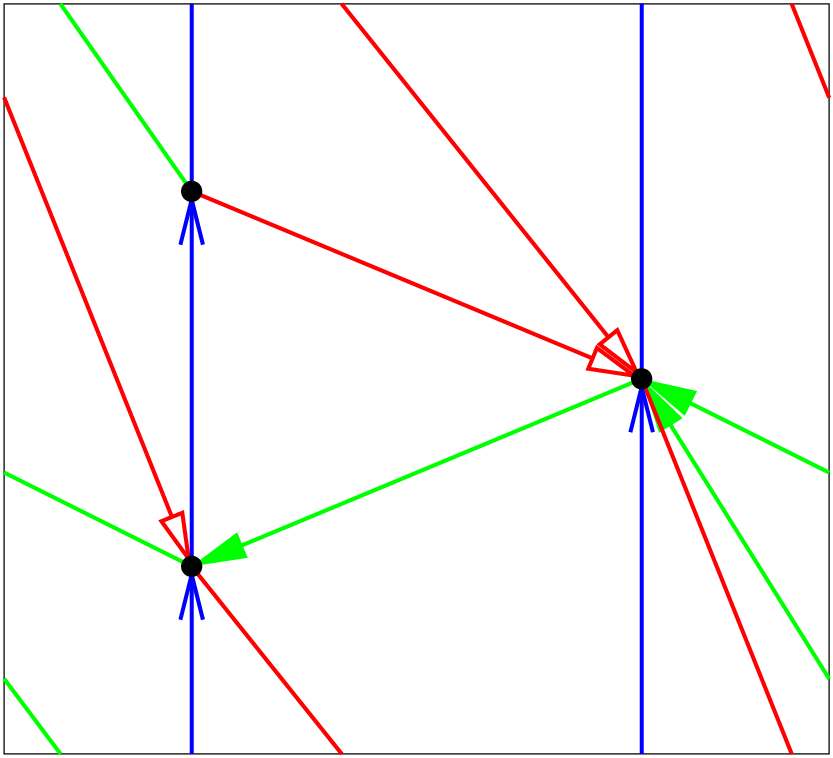}
\caption{A Schnyder wood of a toroidal map and of a toroidal
  triangulation.}
\label{fig:tore-primal}
\end{figure}

The first and third author already defined Schnyder woods for toroidal
maps in~\cite{GL13}. Our definition is broader as in~\cite{GL13},
there is also a (global) condition on the way monochromatic cycles
intersect. See Section~\ref{sec:crossing} for a discussion on this
property.

Figure~\ref{fig:doubletorus} is an example of a Schnyder wood on
a triangulation of the double torus. 
The correspondence from Figure~\ref{fig:edgelabeling} immediately gives the
following bijection whose proof is omitted.

\begin{proposition}
\label{prop:bijtore}
If $G$ is a map on a genus $g\geq 1$ orientable surface, then the
generalized Schnyder woods of $G$ are in bijection with the {\sc
  edge}, $\mathbb{N}^*$-{\sc vertex}, $\mathbb{N}^*$-{\sc face} angle
labelings of $G$.
\end{proposition}

The examples in Figures~\ref{fig:tore-primal}
and~\ref{fig:doubletorus} do not have type 0 edges. However, for all
$g\ge 2$, there are genus $g$ maps, with vertex degrees and face
degrees at most five.  Figure~\ref{fig:annoying-maps} depicts how to
construct such maps, for all $g\ge 2$.  For these maps, type 0 edges
are unavoidable.  Indeed, take such a map with an angle labeling that
has only type 1 and type 2 edges. Around a type 1 or type 2 edge there
are exactly three changes of labels, so in total there are exactly $3m$
such changes. As vertices and faces have degree at most five, they are
either of type 0 or 1, hence the number of label changes should be at
most $3n +3f$. Thus, $3m \le 3n + 3f$, which contradicts Euler's
formula for $g\ge 2$.
Furthermore, note that the maps described in
Figure~\ref{fig:annoying-maps}, as well as their dual maps, are 3-connected.
Actually they can be modified to be 4-connected and of arbitrary
large face-width.

\begin{figure}[!h]
\center
\begin{tabular}{cc}
\includegraphics[scale=0.22]{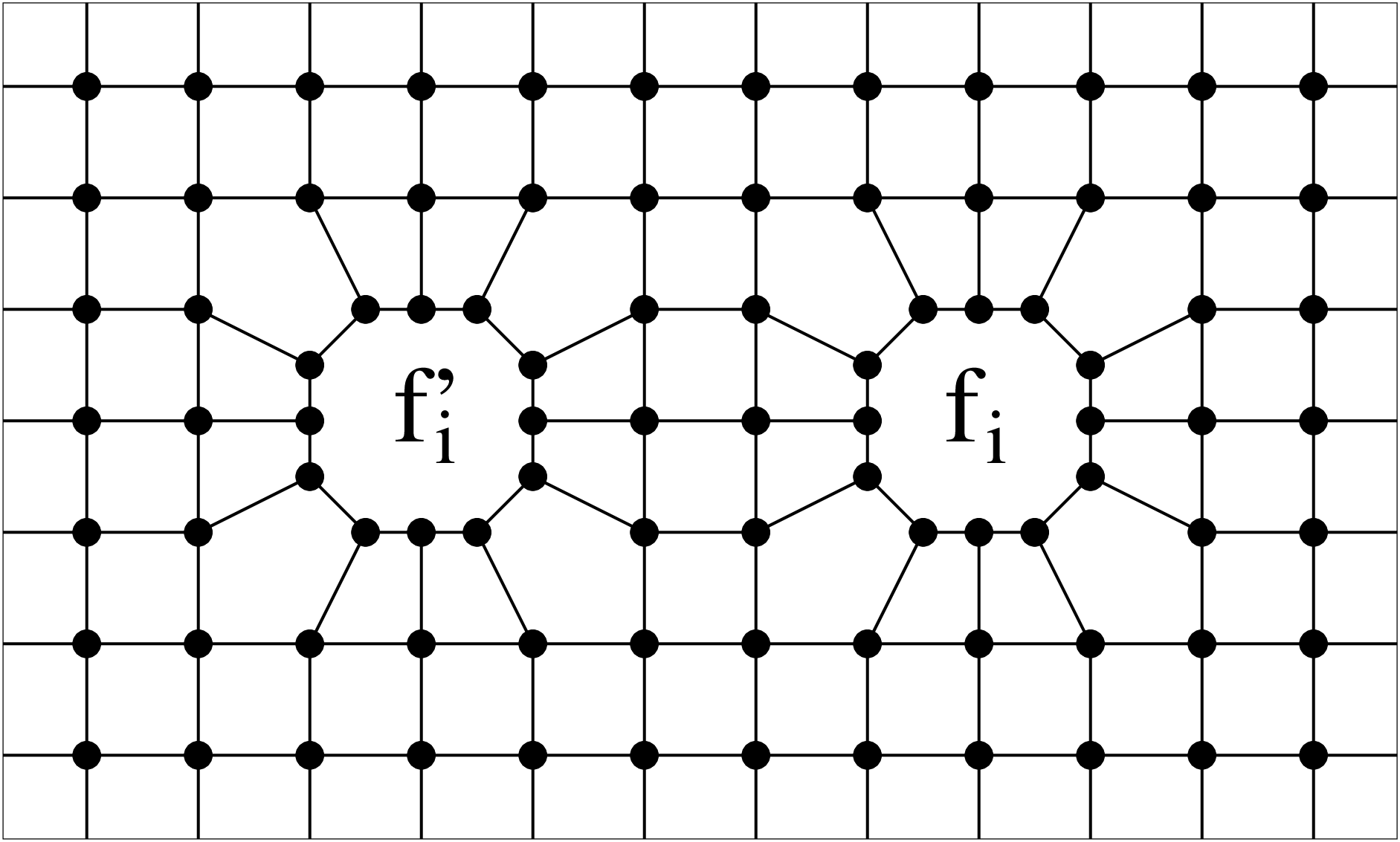} &
\includegraphics[scale=0.28]{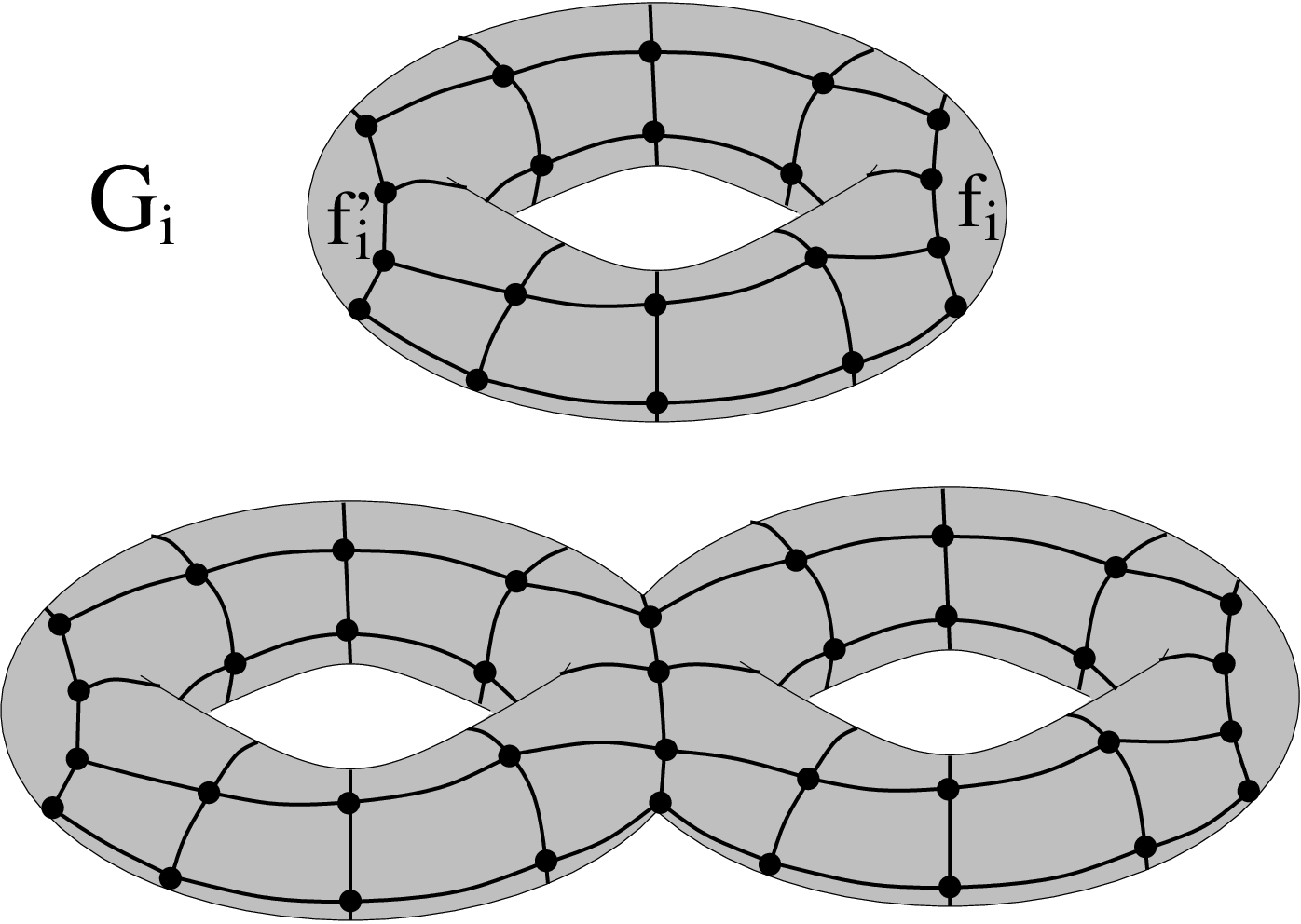}
\end{tabular}
\caption{A toroidal map $G_i$ with two distinguished faces, $f_i$ and
  $f'_i$. Take $g$ copies $G_i$ with $1\le i\le g$ and glue them by
  identifying $f_i$ and $f'_{i+1}$ for all $1\le i<g$. Faces $f_1$ and
$f'_g$ are filled to have only vertices and faces of degree at most five.}
\label{fig:annoying-maps}
\end{figure}

An orientation and coloring of the edges corresponding to an {\sc
  edge}, $\mathbb{N}^*$-{\sc vertex}, $\mathbb{N}^*$-{\sc face} angle
labelings is given for the double-toroidal map of
Figure~\ref{fig:annoying-maps-small}. It contains two edges of type 0
and it is $1$-{\sc vertex} and $1$-{\sc face}. Similarly, one can
obtain {\sc edge}, $\mathbb{N}^*$-{\sc vertex}, $\mathbb{N}^*$-{\sc
  face} angle labelings for any map in Figure~\ref{fig:annoying-maps}.

\begin{figure}[!h]
\center
\includegraphics[scale=0.3]{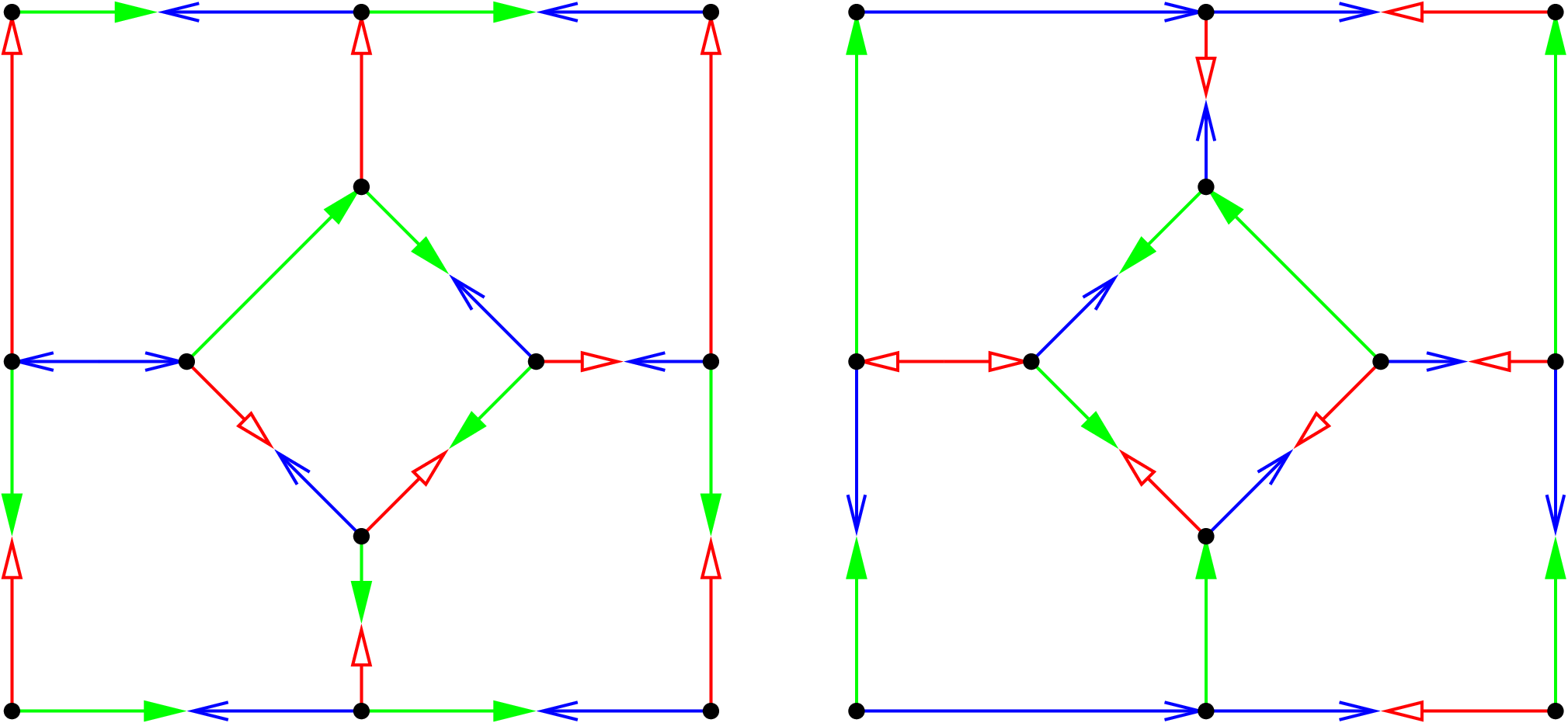}
\caption{An orientation and coloring of the edges of a double-toroidal
  map that correspond to an {\sc edge}, $\mathbb{N}^*$-{\sc vertex},
  $\mathbb{N}^*$-{\sc face} angle labeling. Here, the two parts are
  toroidal and the two central faces are identified (by preserving the
  colors) to obtain a double-toroidal map.}
\label{fig:annoying-maps-small}
\end{figure}

\subsection{Schnyder woods in the universal cover}
\label{sec:universalcover}
In this section we prove some properties of Schnyder woods in the
universal cover.  We refer to~\cite{Massey} for the general theory of
universal covers.  The \emph{universal cover} of the torus (resp. an
orientable surface of genus $g\geq 2$) is a surjective mapping $p$
from the plane (resp. the open unit disk) to the surface that is locally a
homeomorphism.  The universal cover of the torus is obtained by
replicating a flat representation of the torus to tile the plane.
Figure~\ref{fig:universalcover} shows how to obtain the universal
cover of the double torus.  The key property is that a closed curve on
the surface corresponds to a closed curve in the universal cover if
and only if it is contractible.

\begin{figure}[!h]
\center
\includegraphics[scale=1.2]{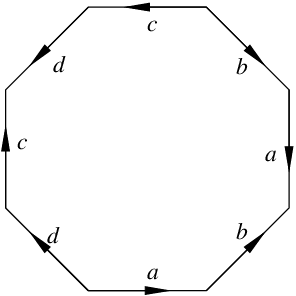}
\hspace{2em}
\includegraphics[scale=0.75]{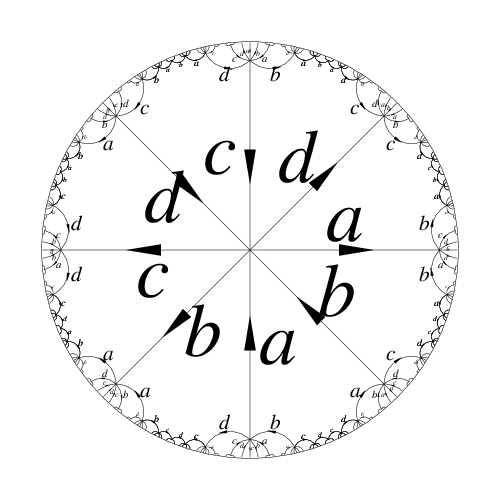}
\caption{Canonical representation and universal cover of the double torus
  (source : Yann Ollivier
  \texttt{http://www.yann-ollivier.org/maths/primer.php}).}
\label{fig:universalcover}
\end{figure}

Universal covers can be used to represent a map on an orientable surface
as an infinite planar map. Any property of the map can be lifted to
its universal cover, as long as it is defined locally. Thus universal
covers are an interesting tool for the study of Schnyder labelings
since all the definitions we have given so far are purely local.

Consider a map $G$ on a genus $g\geq 1$ orientable surface.  Let
$G^\infty$ be the infinite planar map drawn on the universal cover and
defined by $p^{-1}(G)$.

We need the following general lemma concerning universal covers:

\begin{lemma}
\label{lem:finitecc}
Suppose that for a finite set of vertices $X$ of $G^\infty$, the graph
$G^\infty\setminus X$ is not connected. Then $G^\infty\setminus X$ has
a finite connected component.
\end{lemma}

\begin{proof}
  Suppose the lemma is false and $G^\infty\setminus X$ is not
  connected and has no finite component. Then it has a face bounded by
  an infinite number of vertices. As $G$ is finite, the vertices of
  $G^\infty$ have bounded degree. Putting back the vertices of $X$, a
  face bounded by an infinite number of vertices would remain.  The
  border of this face does not correspond to a contractible cycle of
  $G$, a contradiction with $G$ being a map.
\end{proof}

Recall that a graph is \emph{$k$-connected} if it has at least $k+1$
vertices and if it remains connected after removing any $k-1$ vertices.
Extending the notion of essentially 2-connectedness defined in
\cite{MR98} for the toroidal case, we say that $G$ is
\emph{essentially $k$-connected} if $G^\infty$ is $k$-connected. Note
that the notion of being essentially $k$-connected is different from $G$
being $k$-connected. There is no implications in any direction. But note
that since $G$ is a map, it is essentially $1$-connected.

Suppose now that $G$ is given with a Schnyder wood (i.e. an {\sc
  edge}, $\mathbb{N}^*$-{\sc vertex}, $\mathbb{N}^*$-{\sc face} angle
labeling by Proposition~\ref{prop:bijtore}).  Consider the orientation
and coloring of the edges of $G^\infty$ corresponding to the Schnyder
wood of $G$.

Let $G^\infty_i$ be the directed graph induced by the edges of color
$i$ of $G^\infty$. This definition includes edges that are
half-colored $i$, and in this case, the edges get only the direction
corresponding to color $i$.  The graph $(G^\infty_i)^{-1}$ is the
graph obtained from $G^\infty_i$ by reversing all its edges.  The
graph $G^\infty_i\cup (G^\infty_{i-1})^{-1}\cup (G^\infty_{i+1})^{-1}$
is obtained from the graph $G$ by orienting edges in one or two
directions depending on whether this orientation is present in
$G^\infty_i$, $(G^\infty_{i-1})^{-1}$ or $(G^\infty_{i+1})^{-1}$.
Similarly to what happens for planar Schnyder woods, we have the
following:

\begin{lemma}
  \label{lem:nodirectedcycle}
The graph $G^\infty_i\cup (G^\infty_{i-1})^{-1}\cup
(G^\infty_{i+1})^{-1}$ does not contain directed cycle.
\end{lemma}

\begin{proof}
  Suppose there is a directed cycle in $G^\infty_i\cup
  (G^\infty_{i-1})^{-1}\cup (G^\infty_{i+1})^{-1}$.  Let $C$ be such a
  cycle containing the minimum number of faces in the map $D$ with
  border $C$. Suppose by symmetry that $C$ turns around $D$
  counterclockwisely. Every vertex of $D$ has at least one outgoing
  edge of color $i+1$ in $D$.  So there is a cycle of color $(i+1)$ in
  $D$ and this cycle is $C$ by minimality of $C$.  Every vertex of $D$
  has at least one outgoing edge of color $i$ in $D$. So, again by
  minimality of $C$, the cycle $C$ is a cycle of color $i$.  Thus all
  the edges of $C$ are oriented in color $i$ counterclockwisely and in
  color $i+1$ clockwisely.

  By the definition of Schnyder woods, there is no face the boundary
  of which is a monochromatic cycle, so $D$ is not a face.  Let $vx$
  be an edge in the interior of $D$ that is outgoing for $v$. The
  vertex $v$ can be either in the interior of $D$ or in $C$ (if $v$
  has more than three outgoing arcs). In both cases, $v$ has
  necessarily an edge $e_i$ of color $i$ and an edge $e_{i+1}$ of
  color $i+1$, leaving $v$ and in the interior of $D$.  Consider
  $W_i(v)$ (resp. $W_{i+1}(v)$) a monochromatic walk starting from
  $e_i$ (resp. $e_{i+1}$), obtained by following outgoing edges of
  color $i$ (resp. $i+1$).  By minimality of $C$ those walks are not
  contained in $D$. We hence have that $W_i(v)\setminus v$ and
  $W_{i+1}(v)\setminus v$ intersect $C$. Thus each of these walks
  contains a non-empty subpath from $v$ to $C$. The union of these two
  paths, plus a part of $C$ contradicts the minimality of $C$.
\end{proof}

Let $v$ be a vertex of $G^\infty$. For each color $i$, vertex $v$ is
the starting vertex of some walks of color $i$, we denote the union of
these walks by $P_i(v)$. Every vertex has at least one outgoing edge of
color $i$ and the set $P_i(v)$ is obtained by following all these edges
of color $i$ starting from $v$. Note that for some vertices $v$, $P_i(v)$ may
consist of a single walk. It is the case when $v$ cannot reach a
vertex of outdegree six or more.

\begin{lemma}
  \label{lem:nocommon}
  For every vertex $v$ and color $i$, the two graphs
  $P_{i-1}(v)$ and $P_{i+1}(v)$ intersect only on $v$.
\end{lemma}

\begin{proof}
  If $P_{i-1}(v)$ and $P_{i+1}(v)$ intersect on two vertices, then
  $G^\infty_{i-1}\cup (G_{i+1}^{^\infty})^{-1}$ contains a cycle,
  contradicting Lemma~\ref{lem:nodirectedcycle}.
\end{proof}

Now  we can prove the following:

\begin{theorem}
\label{lem:conjessentially}
If a map $G$ on a genus $g\geq 1$ orientable surface admits an {\sc
  edge}, $\mathbb{N}^*$-{\sc vertex}, $\mathbb{N}^*$-{\sc face} angle
labeling, then $G$ is essentially 3-connected.
\end{theorem}

\begin{proof}
  Towards a contradiction, suppose that there exist two vertices $x,y$ of
  $G^\infty$ such that $G'=G^\infty\setminus\{x,y\}$ is not
  connected. Then, by Lemma~\ref{lem:finitecc}, the graph $G'$ has a
  finite connected component $R$. Let $v$ be a vertex of $R$. By
  Lemma~\ref{lem:nodirectedcycle}, for $0\leq i \leq 2$, the graph
  $P_{i}(v)$ does not lie in $R$ so it intersects either $x$ or
  $y$. So for two distinct colors $i,j$, the two graphs
  $P_{i}(v)$ and $P_{j}(v)$ intersect in a vertex distinct from $v$,
  a contradiction to Lemma~\ref{lem:nocommon}.
\end{proof}

\subsection{Conjectures on the existence of Schnyder woods}
\label{sec:conjectureexistence}

Proving that every triangulation on a genus $g\geq 1$ orientable
surface admits a 1-{\sc edge} angle labeling would imply the following
theorem of Bar\'at and Thomassen~\cite{BT06}:

\begin{theorem}[\cite{BT06}]
\label{th:barat}
A simple triangulation on a genus $g\geq 1$ orientable surface admits
an orientation of its edges such that every vertex has outdegree
divisible by three.
\end{theorem}

Recently, Theorem~\ref{th:barat} has been improved
by Albar,  the first author, and the second author~\cite{AGK14}:

\begin{theorem}[\cite{AGK14}]
\label{th:AGK}
A simple triangulation on a genus $g\geq 1$ orientable
surface admits an orientation of its edges such that every vertex has
outdegree at least three, and divisible by three.
\end{theorem}

Note that Theorems~\ref{th:barat} and~\ref{th:AGK} are proved only in
the case of simple triangulations (i.e. no loops and no multiple
edges). We believe them to be true also for non-simple triangulations
without contractible loops nor contractible double edges.

Theorem~\ref{th:AGK} suggests the existence of 1-{\sc edge} angle
labelings with no sinks, i.e. 1-{\sc edge}, $\mathbb{N}^*$-{\sc
  vertex} angle labelings. One can easily check that in a
triangulation, a 1-{\sc edge} angle labeling is also 1-{\sc
  face}. Thus we can hope that a triangulation on a genus $g\geq 1$
orientable surface admits a 1-{\sc edge}, $\mathbb{N}^*$-{\sc vertex},
1-{\sc face} angle labeling. Note that a 1-{\sc edge}, 1-{\sc face}
angle labeling of a map implies that faces are triangles. So we
propose the following conjecture, whose ``only if'' part follows from the
previous sentence:

\begin{conjecture}
\label{conjecture}
  A map on a genus $g\geq 1$ orientable surface admits a
  1-{\sc edge}, $\mathbb{N}^*$-{\sc vertex}, 1-{\sc face} angle
  labeling if and only if it is a triangulation.
\end{conjecture}

If true, Conjecture~\ref{conjecture} would strengthen
Theorem~\ref{th:AGK} in two ways. First, it considers more
triangulations (not only simple ones). Second, it requires the
coloring property around vertices.

How about general maps? We propose the
following conjecture, whose ``only if'' part is Theorem~\ref{lem:conjessentially}:

\begin{conjecture}
\label{conjecture2}
A map on a genus $g\geq 1$ orientable surface admits an {\sc edge},
$\mathbb{N}^*$-{\sc vertex}, $\mathbb{N}^*$-{\sc face} angle labeling
if and only if it is essentially 3-connected.
\end{conjecture}

Conjecture~\ref{conjecture2} implies Conjecture~\ref{conjecture} since
for a triangulation every face would be of type 1, and thus every edge
would be of type 1.  Conjecture~\ref{conjecture2} is proved
in~\cite{GL13} for $g=1$ whereas both conjectures are open for
$g\geq 2$.  Section~\ref{sec:toretri} gives a new proof of
Conjecture~\ref{conjecture} for $g=1$ based on the results in
Section~\ref{sec:char}.

\section{Characterization of Schnyder orientations}
\label{sec:char}

\subsection{A bit of homology}
\label{sec:homology}

In the next sections, we need a bit of surface homology of general
maps, which we will discuss now. For a deeper introduction to homology we refer to~\cite{Gib10}.

For the sake of generality, in this subsection we consider that maps
may have contractible cycles of size 1 or 2.  Consider a map
$G=(V,E)$, on an orientable surface of genus $g$, given with an
arbitrary orientation of its edges. This fixed arbitrary orientation
is implicit in all the paper and is used to handle flows.  A
\emph{flow} $\phi$ on $G$ is a vector in $\mb{Z}^{E}$. For any
$e\in E$, we denote by $\phi_e$ the coordinate $e$ of $\phi$.

A \emph{walk} $W$ of $G$ is a sequence of edges with a direction of
traversal such that the ending point of an edge walk is the starting
point of the next edge.  A walk is \emph{closed} if the start and end
vertices coincide. 
A walk has a \emph{characteristic flow} $\phi(W)$ defined by:

$$\phi(W)_e:=\#\text{times }W\text{ traverses } e \text{ forward} - \#\text{times
}\\
W\text{ traverses } e \text{ backward}$$

This definition naturally extends to sets of walks.  From now on we
consider that a set of walks and its characteristic flow are the same
object and by abuse of notation we can write $W$ instead of
$\phi(W)$. We do the same for \emph{oriented subgraphs}, i.e.,
subgraphs that can be seen as a set of walks of unit length.

A \emph{facial walk} is a closed walk bounding a face.  Let $\mc{F}$
be the set of counterclockwise facial walks and let
$\mb{F}=<\phi(\mc{F})>$ be the subgroup of $\mb{Z}^E$ generated by
$\mc{F}$.  Two flows $\phi, \phi'$ are \emph{homologous} if $\phi
-\phi' \in \mb{F}$.  They are \emph{weakly homologous} if $\phi -\phi'
\in \mb{F}$ or $\phi + \phi' \in \mb{F}$.  We say that a flow $\phi$
is $0$-homologous if it is homologous to the zero flow, i.e. $\phi \in
\mb{F}$.

Let $\mc{W}$ be the set of \emph{closed} walks and let
$\mb{W}=<\phi(\mc{W})>$ be the subgroup of $\mb{Z}^E$ generated by
$\mc{W}$.  The group $H(G)=\mb{W}/\mb{F}$ is the \emph{first homology
  group} of $G$. It is well-known that $H(G)$ only depends on the
genus of the map, and actually it is isomorphic to $\mb{Z}^{2g}$.

A set $\{B_1,\ldots,B_{2g}\}$ of (closed) walks of $G$ is said to be a
\emph{basis for the homology} if the equivalence classes of their
characteristic vectors $([\phi(B_1)],\ldots,[\phi(B_{2g})])$ generate
$H(G)$.  Then for any closed walk $W$ of $G$, we have
$W=\sum_{F\in\mc{F}}\lambda_FF+\sum_{1\leq i\leq 2g}\mu_iB_i$ for some
$\lambda\in\mathbb{Z}^\mc{F},\mu\in\mathbb{Z}^{2g}$. Moreover one of the
$\lambda_F$ can be set to zero (and then all the other coefficients
are unique).  Indeed, for any map, there exists a set of cycles that
forms a basis for the homology and it is computationally easy to
build. A possible way is by considering a spanning tree $T$ of $G$,
and a spanning tree $T^*$ of $G^*$ that contains no edges dual to $T$.
By Euler's formula, there are exactly $2g$ edges in $G$
that are not in $T$ nor dual to edges of $T^*$. Each of these $2g$ edges forms a unique
cycle with $T$. It is not hard to see that this set of cycles forms a basis for the
homology.

The edges of the dual $G^*$ of $G$ are oriented such that the dual
$e^*$ of an edge $e$ of $G$ goes from the face on the right of $e$ to
the face on the left of $e$.  Let $\mc{F}^*$ be the set of
counterclockwise facial walks of $G^*$.  Consider
$\{B^*_1,\ldots,B^*_{2g}\}$ a set of closed walks of $G^*$ that form a
basis for the homology. Let $p$ and $d$ be flows of $G$ and $G^*$,
respectively. We define the following: $$\beta(p,d)=\sum_{e\in G}p_e
d_{e^*}$$ Note that $\beta$ is a bilinear function.

\begin{lemma}\label{lm:homologous}
Given  two flows $\phi,\phi'$ of $G$, the following properties are
equivalent to each other:

\begin{enumerate}
\item The two flows $\phi, \phi'$ are homologous.
\item For any closed walk $W$ of $G^*$ we have
  $\beta(\phi,W)=\beta(\phi',W)$.
\item For any $F\in \mc{F^*}$, we have $\beta(\phi,F)=\beta(\phi',F)$, and,
  for any $1\le i\le 2g$, we have $\beta(\phi,B^*_i)=\beta(\phi',B^*_i)$.
\end{enumerate}
\end{lemma}

\begin{proof}
  $(1. \Longrightarrow 3.)$ Suppose that $\phi, \phi'$ are
  homologous. Then we have
  $\phi-\phi'=\sum_{F\in\mc{F}}\lambda_FF$ for some
  $\lambda\in\mathbb{Z}^\mc{F}$. It is easy to see that, for any
  closed walk $W$ of $G^*$, a facial walk $F\in\mc{F}$ satisfies
  $\beta(F,W)=0$, so $\beta(\phi,W)=\beta(\phi',W)$ by linearity of
  $\beta$.
 
  $(3. \Longrightarrow 2.)$ Suppose that for any $F\in \mc{F^*}$, we
  have $\beta(\phi,F)=\beta(\phi',F)$, and, for any $1\le i\le 2g$, we
  have $\beta(\phi,B^*_i)=\beta(\phi',B^*_i)$. Let $W$ be any closed
  walk of $G^*$. We have
  $W=\sum_{F\in\mc{F}^*}\lambda_FF+\sum_{1\leq i\leq 2g}\mu_iB^*_i$
  for some $\lambda\in\mathbb{Z}^\mc{F},\mu\in\mathbb{Z}^{2g}$. Then by
  linearity of $\beta$ we have   $\beta(\phi,W)=\beta(\phi',W)$.

  $(2. \Longrightarrow 1.)$ Suppose
  $\beta(\phi,W)=\beta(\phi',W)$ for any closed walk $W$ of $G^*$. Let
  $z=\phi-\phi'$.  Thus $\beta(z,W)=0$ for any closed walk $W$ of
  $G^*$.  We label the faces of $G$ with elements of $\mathbb{Z}$ as
  follows. Choose an arbitrary face $F_0$ and label it $0$. Then,
  consider any face $F$ of $G$ and a path $P_{F}$ of $G^*$ from $F_0$
  to $F$. Label $F$ with $\ell_F=\beta(z,P_{F})$. Note that the label
  of $F$ is independent from the choice of $P_{F}$. Indeed, for any
  two paths $P_1,P_2$ from $F_0$ to $F$, we have $P_1-P_2$ is a closed
  walk, so $\beta(z,P_1-P_2)=0$ and thus $\beta(z,P_1)=\beta(z,P_2)$.
  Let us show that $z=\sum_{F\in\mc{F}} \ell_F \phi(F)$.

\begin{alignat*}{3}
  \sum_{F\in\mc{F}} \ell_F \phi(F) &=\sum_{e\in G}
  \left(\ell_{F_2}-\ell_{F_1}\right) \phi(e) &
  \text{(face $F_2$ is on the left of $e$ and $F_1$ on the right)}\\
  &=\sum_{e\in G}
  \left(\beta(z,P_{F_2})-\beta(z,P_{F_1})\right)
  \phi(e) &\text{(definition of $\ell_F$)}\\
  &=\sum_{e\in G} \beta(z,P_{F_2}-P_{F_1}) \phi(e)  &\text{(linearity of $\beta$)}\\
  &=\sum_{e\in G} \beta(z,e^*) \phi(e) &\text{($P_{F_1}+e^*-P_{F_2}$ is a closed walk)}\\
  & =\sum_{e\in G} \left(\sum_{e'\in G} z_{e'} \phi(e^*)_{e'^*}\right)
  \phi(e) &\text{(definition of $\beta$)}\\
  &=\sum_{e\in G} z_e \phi(e)\\
  &=z
\end{alignat*}

So $z\in  \mb{F}$ and thus $\phi, \phi'$ are homologous.
\end{proof}

\subsection{General characterization}
\label{sec:gencharacterization}
 
Consider a map $G$ on an orientable surface of genus $g$.  The mapping
of Figure~\ref{fig:edgelabeling} shows how an {\sc edge} angle labeling of
$G$ can be mapped to an orientation of the edges with edges oriented
in one direction or in two opposite directions.  These edges can be
defined more naturally in the primal-dual-completion of $G$.

The \emph{primal-dual-completion} $\pdc{G}$ is the map obtained from
simultaneously embedding $G$ and $G^*$ such that vertices of $G^*$ are
embedded inside faces of $G$ and vice-versa. Moreover, each edge
crosses its dual edge in exactly one point in its interior, which also
becomes a vertex of $\pdc{G}$.  Hence, $\pdc{G}$ is a bipartite graph
with one part consisting of \emph{primal-vertices} and
\emph{dual-vertices} and the other part consisting of
\emph{edge-vertices} (of degree four). Each face of $\pdc{G}$ is a
quadrangle incident to one primal-vertex, one dual-vertex and two
edge-vertices. Actually, the faces of $\pdc{G}$ are in correspondance
with the angles of $G$. This means that angle labelings of $G$
correspond to face labelings of $\pdc{G}$.

Given $\alpha:V\to \mb{N}$, an orientation of $G$ is an
\emph{$\alpha$-orientation}~\cite{Fel04} if for every vertex $v\in V$ its
outdegree $d^+(v)$ equals $\alpha(v)$.  We call an orientation of
$\pdc{G}$ a \emph{$\bmod_3$-orientation} if it is an
$\alpha$-orientation for a function $\alpha$ satisfying
: $$\alpha(v)\equiv \begin{cases}
  0 \ (\bmod 3) & \text{if }v\text{ is a primal- or dual-vertex}, \\
  1 \ (\bmod 3) & \text{if }v\text{ is an edge-vertex.}\\
\end{cases}
$$

Note that an {\sc edge} angle labeling of $G$ corresponds to a
$\bmod_3$-orientation of $\pdc{G}$, by the mapping of
Figure~\ref{fig:pdc}, where the three types of edges are
represented. Type 0 corresponds to an edge-vertex of outdegree
four. Type 1 and type 2 both correspond to an edge-vertex of outdegree
$1$; in type 1 (resp. type 2) the outgoing edge goes to a
primal-vertex (resp. dual-vertex). In all cases we have $d^+(v) \equiv 1 \ (\bmod 3)$ if
$v$ is an edge-vertex. By Lemma~\ref{lem:EDGElabeling}, the labeling
is also {\sc vertex} and {\sc face}. Thus, $d^+(v) \equiv 0 \ (\bmod 3)$ if $v$
is a primal- or dual-vertex.

\begin{figure}[!h]
\center
\begin{tabular}{ccc}
\includegraphics[scale=0.5]{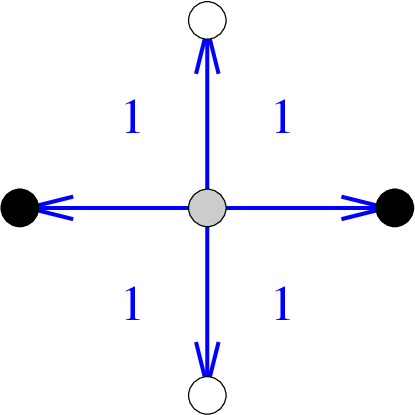} \ \ \  &  \ \ \
\includegraphics[scale=0.5]{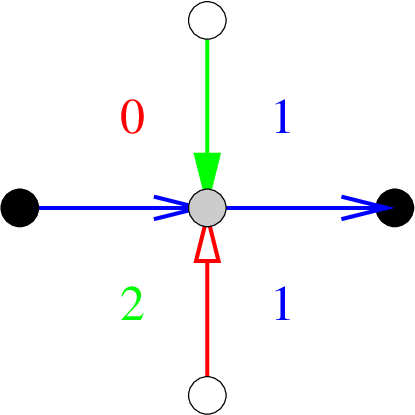} \ \ \  &  \ \ \
\includegraphics[scale=0.5]{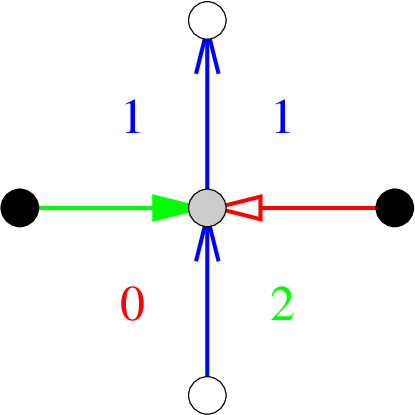} \\
Type 0  \ \ \  &  \ \ \ Type 1  \ \ \  &  \ \ \ Type 2  \\
\end{tabular}
\caption{How to map an {\sc edge} angle labeling to a
  $\bmod_3$-orientation of the primal-dual completion. Primal-vertices
  are black, dual-vertices are white and edge-vertices are gray. This
  serves as a convention for the other figures.}
\label{fig:pdc}
\end{figure}

As mentioned earlier, De Fraysseix and Ossona de
Mendez~\cite{FO01} give a bijection between internal 3-orientations
and Schnyder woods of planar triangulations. Felsner~\cite{Fel04}
generalizes this result for planar Schnyder woods and orientations of
the primal-dual completion having prescribed out-degrees. The
situation is more complicated in higher genus
(see Figure~\ref{fig:orientation}). It is not enough to prescribe
outdegrees in order to characterize orientations corresponding to {\sc
  edge} angle labelings.

We call an orientation of $\pdc{G}$ corresponding to an {\sc edge}
angle labeling of $G$ a \emph{Schnyder orientation}.  In this section
we characterize which orientations of $\pdc{G}$ are Schnyder
orientations.

Consider an orientation of the primal-dual completion $\pdc{G}$.  Let
$\Out=\{(u,v)\in E(\pdc{G})\mid v\text{ is an edge-vertex}\}$, i.e.
the set of edges of $\pdc{G}$ which are going from a primal- or
dual-vertex to an edge-vertex. We call these edges \emph{out-edges}.
For $\phi$ a flow of the dual of the primal-dual completion
$\pdc{G}^*$, we define $\delta(\phi)=\beta(\Out,\phi)$.  More
intuitively, if $W$ is a walk of $\pdc{G}^*$, then:
$$
\begin{array}{ll}
  \delta(W)  = &  \ \ \#\text{out-edges crossing }W\text{
    from left to right}\\
  & -\#\text{out-edges crossing }W\text{ from right to
    left}.  
\end{array}
$$

The bilinearity of $\beta$ implies the linearity of $\delta$.
The following lemma gives a necessary and sufficient condition for an
orientation to be a Schnyder orientation.

\begin{lemma}
 \label{lem:charforall}
 An orientation of $\pdc{G}$ is a Schnyder orientation if and only if
 any closed walk $W$ of $\pdc{G}^*$ satisfies
 $\delta (W) \equiv 0 \ (\bmod 3)$.
\end{lemma}

\begin{proof} $(\Longrightarrow)$ Consider an {\sc edge} angle
  labeling $\ell$ of $G$ and the corresponding Schnyder orientation
  (see Figure~\ref{fig:pdc}).
  Figure~\ref{fig:figdelta} illustrates how $\delta$ counts the
  variation of the label when going from one face of $\pdc{G}$ to
  another face of $\pdc{G}$ . The represented cases
  correspond to a walk $W$ of $\pdc{G}^*$ consisting of just one
  edge. If the edge of $\pdc{G}$ crossed by $W$ is not an out-edge,
  then the two labels in the face are the same and $\delta(W) =0$. If
  the edge crossed by $W$ is an out-edge, then the labels differ by
  one. If $W$ is going counterclockwise around a primal- or
  dual-vertex, then the label increases by $1 \ (\bmod 3)$ and
  $\delta(W)=1$. If $W$ is going clockwise around a primal- or
  dual-vertex then the label decreases by $1 \ (\bmod 3)$ and
  $\delta(W)=-1$. One can check that this is consistent with all the
  edges depicted in Figure~\ref{fig:pdc}.  Thus for any walk $W$ of
  $\pdc{G}^*$ from a face $F$ to a face $F'$, the value of $\delta(W)
  \ (\bmod 3)$ is equal to $\ell(F')-\ell(F) \ (\bmod 3)$.  Thus if $W$ is a
  closed walk then $\delta(W)\equiv 0 \ (\bmod 3)$.

\begin{figure}[!h]
\center
\begin{tabular}{cc}
\includegraphics[scale=0.5]{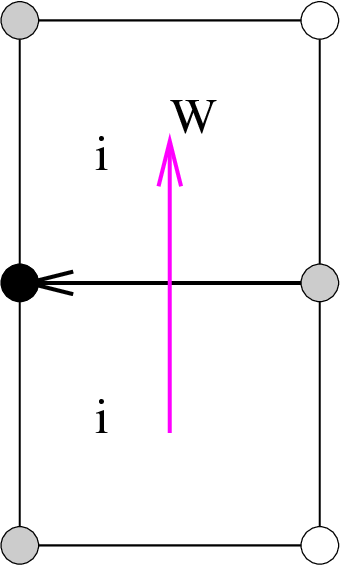} \ \ \  &  \ \ \
\includegraphics[scale=0.5]{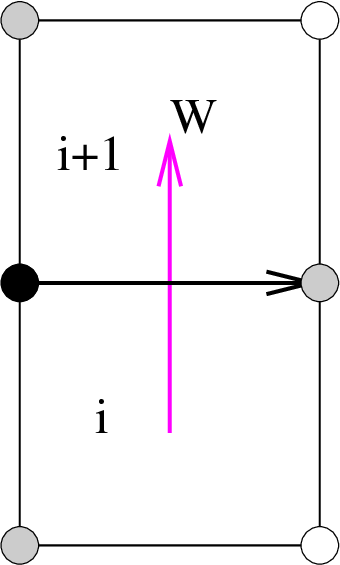} \\
$\delta(W) =0$ \ \ \  &  \ \ \ $\delta(W)=1$ \\
\end{tabular}
\caption{How $\delta$ counts the variation of the labels.}
\label{fig:figdelta}
\end{figure}

$(\Longleftarrow)$ Consider an orientation of $\pdc{G}$ such that
any closed
walk $W$ of $\pdc{G}^*$ satisfies $\delta (W) \equiv 0 \ (\bmod 3)$.  Pick any
face $F_0$ of $\pdc{G}$ and label it $0$. Consider any face $F$ of
$\pdc{G}$ and a path $P$ of $\pdc{G}^*$ from $F_0$ to $F$. Label $F$
with the value $\delta (P)\bmod 3$. Note that the label of $F$ is
independent from the choice of $P$ as for any two paths $P_1, P_2$
going from $F_0$ to $F$, we have $\delta (P_1) \equiv \delta (P_2) \ (\bmod 3)$
since $\delta (P_1 - P_2) \equiv 0 \ (\bmod 3)$ as $P_1- P_2$ is a closed walk.

Consider an edge-vertex $v$ of $\pdc{G}$ and a walk $W$ of $\pdc{G}^*$
going \cw around $v$. By assumption  $\delta (W) \equiv 0 \ (\bmod
3)$ and $d(v)=4$ so $d^+(v)\equiv 1 \ (\bmod 3)$.
 One can check (see Figure~\ref{fig:pdc}) that around an
edge-vertex $v$ of outdegree four, all the labels are the same and thus
$v$ corresponds to an edge of $G$ of type 0. One can also check that
around an edge-vertex $v$ of outdegree $1$, the labels are in clockwise
order, $i-1$, $i$, $i$, $i+1$ for some $i$ in $\{0,1,2\}$ where the
two faces with the same label are incident to the outgoing edge of
$v$. Thus, $v$ corresponds to an edge of $G$ of type 1 or 2 depending
on the outgoing edge reaching a primal- or a
dual-vertex. So the obtained labeling of the faces of $\pdc{G}$
corresponds to an {\sc edge} angle labeling of $G$ and the considered
orientation is a Schnyder orientation.
\end{proof}

We now study properties of $\delta$ w.r.t homology in order to
simplify the condition of Lemma~\ref{lem:charforall}.  Let
$\hat{\mc{F}^*}$ be the set of counterclockwise facial walks of
$\pdc{G}^*$.

\begin{lemma}
\label{lem:facedelta0}
In a $\bmod_3$-orientation of $\pdc{G}$, any $F\in\hat{\mc{F}^*}$
satisfies $\delta(F)\equiv 0 \ (\bmod 3)$.
\end{lemma}
\begin{proof}
  If $F$ corresponds to an edge-vertex $v$ of $\pdc{G}$, then $v$ has
  degree exactly four and outdegree one or four by definition of
  $\bmod_3$-orientations. So there are exactly zero or three out-edges
  crossing $F$ from right to left, and $\delta(F)\equiv 0 \ (\bmod 3)$.

  If $F$ corresponds to a primal- or dual-vertex $v$, then $v$ has
  outdegree $0 \ (\bmod 3)$ by definition of $\bmod_3$-orientations. So
  there are exactly $0 \ (\bmod 3)$ out-edges crossing $F$ from left to
  right, and $\delta(F)\equiv 0 \ (\bmod 3)$.
\end{proof}

\begin{lemma}
\label{lem:basedelta}
In a $\bmod_3$-orientation of $\pdc{G}$, if $\{B_1,\ldots,B_{2g}\}$ is a
set of cycles
of $\pdc{G}^*$ that forms a basis for the
homology, then for any closed walk $W$ of $\pdc{G}^*$ homologous to
$\mu_1 B_1 + \cdots + \mu_{2g} B_{2g}$, we have
$\delta(W)\equiv \mu_1 \delta(B_1) + \cdots + \mu_{2g} \delta(B_{2g}) \ (\bmod
3)$.
\end{lemma}

\begin{proof} We have
  $W=\sum_{F\in\hat{\mc{F}^*}}\lambda_FF+\sum_{1\leq i\leq
    2g}\mu_iB_i$
  for some $\lambda\in\mathbb{Z}^\mc{F}$.  Then by linearity of $\delta$
  and Lemma~\ref{lem:facedelta0}, the claim follows.
\end{proof}

Lemma~\ref{lem:basedelta} can be used to simplify the condition of
Lemma~\ref{lem:charforall} and show that if $\{B_1,\ldots,B_{2g}\}$ is a
set of cycles
of $\pdc{G}^*$ that forms a basis for the homology,
then an orientation of $\pdc{G}$ is a Schnyder orientation if and only
if it is a $\bmod_3$-orientation such that $\delta(B_{i})\equiv 0 \ (\bmod 3)$,
for all $1\le i\le 2g$. Now, we define a new function $\gamma$
that is used to formulate a similar characterization theorem (see Theorem~\ref{th:characterizationgamma}).

Consider a (not necessarily directed) cycle
$C$ of $G$ together with a direction of traversal. We associate to
$C$ its corresponding cycle in $\pdc{G}$ denoted by $\pdc{C}$.  We
define $\gamma(C)$ by:
$$\gamma (C) = \#\ \text{edges of $\pdc{G}$ leaving $\pdc{C}$ on its right} -
\#\  \text{edges of $\pdc{G}$ leaving $\pdc{C}$ on its left}$$

Since it considers cycles of $\pdc{G}$ instead of walks of
$\pdc{G}^*$, it is easier to deal with parameter $\gamma$ rather than
parameter $\delta$.  However $\gamma$ does not enjoy the same property
w.r.t. homology as $\delta$.  For homology we have to consider walks
as flows, but two walks going several time through a given vertex may
have the same characteristic flow but different $\gamma$.  This
explains why $\delta$ is defined first. Now we adapt the results for
$\gamma$.

The value of $\gamma$ is related to $\delta$ by the next lemmas. Let
$C$ be a cycle of $G$ with a direction of traversal. Let $W_L(C)$ be
the closed walk of $\pdc{G}^*$ just on the left of $C$ and going in
the same direction as $C$ (i.e. $W_L(C)$ is composed of the dual edges
of the edges of $\pdc{G}$ incident to the left of $\pdc{C}$).  Note
that since the faces of $\pdc{G}^*$ have exactly one incident vertex
that is a primal-vertex, walk $W_L(C)$ is in fact a cycle of
$\pdc{G}^*$.  Similarly, let $W_R(C)$ be the cycle of $\pdc{G}^*$ just
on the right of $C$.

\begin{lemma}
\label{lem:gammaequaldelta}
Consider an orientation of $\pdc{G}$ and 
a cycle $C$ of $G$, then $\gamma(C) = \delta (W_L(C)) + \delta (W_R(C))$.
\end{lemma}

\begin{proof}
  We consider the different cases that can occur. An edge that is
  entering a primal-vertex of $\pdc{C}$, is not counting in either
  $\gamma(C),\delta (W_L(C)), \delta (W_R(C))$.  An edge that is
  leaving a primal-vertex of $\pdc{C}$ from its right side (resp. left
  side) is counting $+1$ (resp. $-1$) for $\gamma(C)$ and
  $\delta (W_R(C))$ (resp.  $\delta (W_L(C))$).

  For edges incident to edge-vertices of $\pdc{C}$ both sides have to
  be considered at the same time. Let $v$ be an edge-vertex of
  $\pdc{C}$. Vertex $v$ is of degree four so it has exactly two edges
  incident to $\pdc{C}$ and not on $C$. One of these edges, $e_L$, is on
  the left side of $\pdc{C}$ and dual to an edge of $W_L(C)$. The
  other edge, $e_R$, is on the right side of $\pdc{C}$ and dual to an
  edge of $W_R(C)$. If $e_L$ and $e_R$ are both incoming edges for
  $v$, then $e_R$ (resp. $e_L$) is counting $-1$ (resp. $+1$) for
  $\delta (W_R(C))$ (resp.  $\delta (W_L(C))$) and not counting for
  $\gamma(C)$. If $e_L$ and $e_R$ are both outgoing edges for $v$,
  then $e_R$ and $e_L$ are not counting for both $\delta (W_R(C))$,
  $\delta (W_L(C))$ and sums to zero for $\gamma(C)$. If $e_L$ is
  incoming and $e_R$ is outgoing for $v$, then $e_R$ (resp. $e_L$) is
  counting $0$ (resp. $+1$) for $\delta (W_R(C))$ (resp.
  $\delta (W_L(C))$), and counting $+1$ (resp. $0$) for
  $\gamma(C)$. The last case, $e_L$ is outgoing
   and $e_R$ is incoming, is symmetric and one can see that in the
   four cases we have that $e_L$ and $e_R$ count the same for $\gamma(C)$
  and $\delta (W_L(C)) + \delta (W_R(C))$.  We conclude
  $\gamma(C)=\delta (W_L(C)) + \delta (W_R(C))$.
\end{proof}

\begin{lemma}
\label{lem:deltagammaequ0mod3}
In a $\bmod_3$-orientation of $G$, a cycle $C$ of $G$ satisfies
$$\delta (W_L(C))\equiv 0 \ (\bmod 3) \ \ \text{and}\ \ \delta (W_R(C))\equiv 0 \ (\bmod 3) \iff
\gamma(C) \equiv 0 \ (\bmod 3)
$$ 
\end{lemma}

\begin{proof}
($\Longrightarrow$) Clear by Lemma~\ref{lem:gammaequaldelta}.

($\Longleftarrow$) Suppose that $\gamma(C)\equiv 0 \ (\bmod 3)$.  Let $x_L$
(resp. $y_L$) be the number of edges of $\pdc{G}$ that are dual to edges
of $W_L(C)$, that are outgoing for a primal-vertex of $\pdc{C}$
(resp. incoming for an edge-vertex of $\pdc{C}$). Similarly, let $x_R$
(resp. $y_R$) be the number of edges of $\pdc{G}$ that are dual to edges
of $W_R(C)$, that are outgoing for a primal-vertex of $\pdc{C}$
(resp. incoming for an edge-vertex of $\pdc{C}$).  So
$\delta (W_L(C))=y_L-x_L$ and $\delta (W_R(C))=x_R-y_R$. So by
Lemma~\ref{lem:gammaequaldelta},
$ \gamma(C)= \delta (W_L(C))+\delta (W_R(C))=(y_L+x_R)-(x_L+y_R) \equiv 0
\ (\bmod 3)$.

Let $k$ be the number of vertices of $C$. So $\pdc{C}$ has $k$
primal-vertices, $k$ edge-vertices and $2k$ edges.  Edge-vertices have
outdegree $1 \ (\bmod 3)$ so their total number of outgoing edges on
$\pdc{C}$ is $k+(y_L+y_R) \ (\bmod 3)$.  Primal-vertices have outdegree
$0 \ (\bmod 3)$ so their total number of outgoing edges on $\pdc{C}$ is
$-(x_L+x_R) \ (\bmod 3)$. So in total $2k\equiv k+(y_L+y_R)-(x_L+x_R) \ (\bmod 3)$.
So $(y_L+y_R)-(x_L+x_R)\equiv 0 \ (\bmod 3)$. By combining this with plus
(resp. minus) $(y_L+x_R)-(x_L+y_R) \equiv 0 \ (\bmod 3)$, one obtains that
$2\delta (W_L(C))=2(y_L-x_L)\equiv 0 \ (\bmod 3)$ (resp. 
$2\delta (W_R(C))=2(x_R-y_R)\equiv 0 \ (\bmod 3)$). Since $\delta (W_L(C))$ and
$\delta (W_R(C))$ are integer we obtain $\delta (W_L(C))\equiv 0 \ (\bmod 3)$
and $\delta (W_R(C))\equiv 0 \ (\bmod 3)$.
\end{proof}

Finally we have the following characterization theorem concerning
Schnyder orientations:

\begin{theorem}
\label{th:characterizationgamma}
Consider a map $G$ on an orientable surface of genus $g$. Let
$\{B_1,\ldots,B_{2g}\}$ be a set of cycles of $G$ that forms a basis
for the homology.  An orientation of $\pdc{G}$ is a Schnyder orientation
if and only if it is a $\bmod_3$-orientation such that
$\gamma(B_{i})\equiv 0 \ (\bmod 3)$, for all $1\le i\le 2g$.
\end{theorem}

\begin{proof}
  $(\Longrightarrow)$ Consider an {\sc edge} angle
  labeling $\ell$ of $G$ and the corresponding Schnyder orientation
  (see Figure~\ref{fig:pdc}). Type~0 edges correspond to edge-vertices
  of outdegree four, while type~1 and~2 edges correspond to
  edge-vertices of outdegree $1$. Thus $d^+(v)\equiv 1 \ (\bmod 3)$ if $v$ is an
  edge-vertex.
By Lemma~\ref{lem:EDGElabeling}, the labeling
  is {\sc vertex} and {\sc face}. Thus $d^+(v) \equiv 0 \ (\bmod 3)$ if $v$ is a
  primal- or dual-vertex. So the orientation is a
  $\bmod_3$-orientation. By Lemma~\ref{lem:charforall}, we have
  $\delta (W) \equiv 0 \ (\bmod 3)$ for any closed walk $W$ of $\pdc{G}^*$. So
  we have that $\delta(W_L(B_1)), \ldots, \delta(W_L(B_{2g}))$,
  $\delta(W_R(B_1)), \ldots, \delta(W_R(B_{2g}))$ are all congruent to
  $0 \ (\bmod 3)$.  Thus, by Lemma~\ref{lem:deltagammaequ0mod3}, we have
  $\gamma(B_{i}) \equiv 0 \ (\bmod 3)$, for all $1\le i\le 2g$.

  $(\Longleftarrow)$ Consider a $\bmod_3$-orientation of $G$ such that
  $\gamma(B_{i}) \equiv 0 \ (\bmod 3)$, for all $1\le i\le 2g$.  By
  Lemma~\ref{lem:deltagammaequ0mod3}, we have 
  $\delta(W_L(B_{i})) \equiv 0 \ (\bmod 3)$ for all $1\le i\le 2g$. Moreover
  $\{W_L(B_1),\ldots,W_L(B_{2g})\}$ forms a basis for the homology.  So
  by Lemma~\ref{lem:basedelta}, $\delta (W) \equiv 0 \ (\bmod 3)$ for any
  closed walk $W$ of $\pdc{G}^*$. So the orientation is a Schnyder
  orientation by Lemma~\ref{lem:charforall}.
\end{proof}

The condition of Theorem~\ref{th:characterizationgamma} is easy to
check: choose $2g$ cycles that form a basis for the homology and check
whether $\gamma$ is congruent to $0\bmod 3$ for each of them.

When restricted to triangulations and to edges of type $1$ only, the
defintion of $\gamma$ can be simplified. Consider a triangulation $G$
on an orientable surface of genus $g$ and an orientation of the edges
of $G$.
Figure~\ref{fig:trimap} shows how
 to transform the orientation of $G$ into an orientation of
 $\pdc{G}$. Note that all the edge-vertices have outdegree exactly
 $1$. Furthermore, all the dual-vertices only have outgoing edges and
 since we are considering triangulations they have outdegree exactly
 three.

 \begin{figure}[!h]
 \center
 \begin{tabular}{cc}
 \includegraphics[scale=0.5]{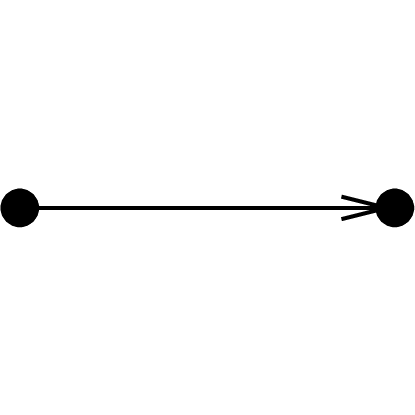} \ \ \  &  \ \ \
 \includegraphics[scale=0.5]{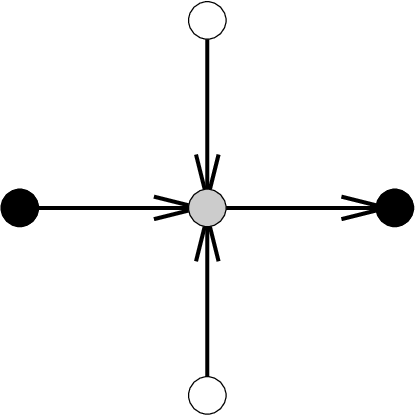} \\
 $G$ \ \ \  &  \ \ \ $\pdc{G}$ \\
 \end{tabular}
 \caption{How to transform an orientation of a triangulation $G$ into
   an orientation of $\pdc{G}$.}
 \label{fig:trimap}
 \end{figure}

Then the definition of $\gamma$ can be
simplified by the following:
$$\gamma (C) = \#\ \text{edges of ${G}$ leaving ${C}$ on its right} -
\#\  \text{edges of ${G}$ leaving ${C}$ on its left}$$

Note that comparing to the general definition of $\gamma$, only the symbols
$\ \pdc{ }\ $ have been removed.

The orientation of the toroidal triangulation on the left of
Figure~\ref{fig:orientation} is an example of a 3-orientation of a
toroidal triangulation where some non contractible cycles have value
$\gamma$ not congruent to $0 \bmod 3$.  The value of $\gamma$ for the
three loops is $2, 0$ and $-2$. This explains why this orientation does
not correspond to a Schnyder wood. On the contrary, on the right of
the figure, the three loops have $\gamma$ equal to $0$ and we have a
Schnyder wood.

\section{Structure of Schnyder orientations}
\label{sec:stru}
\subsection{Transformations between Schnyder orientations}
\label{sec:transformations}

We investigate the structure of the set of Schnyder orientations of a
given graph. For that purpose we need some definitions that are given
on a general map $G$ and then applied to $\pdc{G}$.

Consider a map $G$ on an orientable surface of genus $g$.  Given two
orientations $D$ and $D'$ of $G$, let $D\setminus D'$ denote the
subgraph of $D$ induced by the edges that are not oriented as in $D'$.

An oriented subgraph $T$ of $G$ is \emph{partitionable} if its edge
set can be partitioned into three sets $T_0$, $T_1$, $T_2$ such that
all the $T_i$ are pairwise homologous, i.e. $T_i-T_j\in\mb{F}$ for
$i,j\in\{0,1,2\}$.  An oriented subgraph $T$ of $G$ is called a
\emph{topological Tutte-orientation} if $\beta(T,W)\equiv 0 \ (\bmod 3)$ for
every closed walk $W$ in $G^*$ (more intuitively, the number of edges
crossing $W$ from left to right minus the number of those crossing $W$
from right to left is divisible by three).

The name ``topological Tutte-orientation'' comes from the fact that an
oriented graph $T$ is called a \emph{Tutte-orientation} if the difference of outdegree and indegree
is divisible by three, i.e. $d^+(v)-d^-(v)\equiv 0 \ (\bmod 3)$, for every vertex $v$.  So a topological
Tutte-orientation is a Tutte orientation, since the latter requires
the condition of the topological Tutte orientation only for the walks
$W$ of $G^*$ going around a vertex $v$ of $G$.

The notions of partitionable and topological
Tutte-orientation are equivalent:

\begin{lemma}\label{lm:topologicalTutte}
  An oriented subgraph of $G$ is partitionable if and only if it is a
  topological Tutte-orientation.
\end{lemma}
\begin{proof}
  $(\Longrightarrow)$ If $T$ is partitionable, then by definition it
  is the disjoint union of three homologous edge sets $T_0$, $T_1$,
  and $T_2$. Hence by Lemma~\ref{lm:homologous}, $\beta(T_0,W) =
  \beta(T_1,W) = \beta(T_2,W)$ for any closed walk $W$ of $G^*$.  By
  linearity of $\beta$ this implies that $\beta(T,W)\equiv 0 \ (\mod 3)$ for
  any closed walk $W$ of $G^*$. So $T$ is a topological
  Tutte-orientation.

  $(\Longleftarrow)$ Let $T$ be a topological Tutte-orientation of
  $G$, i.e. $\beta(T,W)\equiv 0 \ (\mod 3)$ for any closed walk $W$ of $G^*$.
  In the following, \emph{$T$-faces} are the faces of $T$ considered as
  an embedded graph. Note that $T$-faces are not necessarily disks.
  Let us introduce a $\{0,1,2\}$-labeling of the $T$-faces. Label an
  arbitrary $T$-face $F_0$ by $0$. For any $T$-face $F$, find a path
  $P$ of $G^*$ from $F_0$ to $F$. Label $F$ with $\beta(T,P) \ (\bmod
  3)$. Note that the label of $F$ is independent from the choice of $P$
  by our assumption on closed walks.  For $0\leq i \leq 2$, let $T_i$
  be the set of edges of $T$ with two incident $T$-faces labeled
  $i-1$ and $i+1$.  Note that an edge of $T_i$ has label $i-1$ on its
  left and label $i+1$ on its right.  The sets $T_i$ form a partition
  of the edges of $T$.  Let $\mc{F}_i$ be the counterclockwise facial
  walks of $G$ that are in a $T$-face labeled $i$.  We have
  $\phi(T_{i+1})-\phi(T_{i-1})=\sum_{F\in \mc{F}_i}\phi(F)$, so the
  $T_i$ are homologous.
\end{proof}

Let us refine the notion of partitionable.  Denote by $\mathcal{E}$
the set of \emph{oriented Eulerian subgraphs} of ${G}$ (i.e. the
oriented subgraphs of ${G}$ where each vertex has the same in- and
out-degree).  Consider a partitionable oriented subgraph $T$ of $G$,
with edge set partition $T_0$, $T_1$, $T_2$ having the same homology.
We say that $T$ is \emph{Eulerian-partitionable} if $T_i\in\mc{E}$ for
all $0\leq i \leq 2$. Note that if $T$ is Eulerian-partitionable then
it is Eulerian.  Note that an oriented subgraph $T$ of $G$ that is
$0$-homologous is also Eulerian and thus Eulerian-partitionable (with
the partition $T,\emptyset,\emptyset$). 

We now investigate the structure of Schnyder orientations.  For that
purpose, consider a map $G$ on an orientable surface of genus $g$ and
apply the above definitions and results to orientations of $\pdc{G}$.

Let $D,D'$ be two orientations of $\pdc{G}$ such that $D$ is a
Schnyder orientation and $T=D\setminus D'$. Let
$\Out=\{(u,v)\in E(D)\mid v\text{ is an edge-vertex}\}$. Similarly,
let $\Out'=\{(u,v)\in E(D')\mid v\text{ is an edge-vertex}\}$. Note
that an edge of $T$ is either in $\Out$ or in $\Out'$, so
$\phi(T)=\phi(\Out)- \phi(\Out')$. By Lemma~\ref{lem:charforall}, for
any closed walk $W$ of $\pdc{G}^*$, $\beta(\Out,W)\equiv 0 \ (\bmod 3)$. The
three following lemmas give necessary and sufficient conditions on $T$
for $D'$ being a Schnyder orientation.

\begin{lemma}\label{thm:EDGEtransform}
  $D'$ is a Schnyder orientation if and only if $T$ is
  partitionable.
\end{lemma}

\begin{proof}
  Let $D'$ is a Schnyder orientation.  By Lemma~\ref{lem:charforall},
  this is equivalent to the fact that for any closed walk $W$ of
  $\pdc{G}^*$, we have $\beta (\Out',W) \equiv 0 \ (\bmod 3)$. Since
  $\beta(\Out,W)\equiv 0 \ (\bmod 3)$, this is equivalent to the fact that for
  any closed walk $W$ of $\pdc{G}^*$, we have
  $\beta (T,W) \equiv 0 \ (\bmod 3)$. Finally, by
  Lemma~\ref{lm:topologicalTutte} this is equivalent to $T$ being
  partitionable.
\end{proof}

\begin{lemma}\label{th:EDGEeulerian}
  $D'$ is a Schnyder orientation having the same outdegrees as $D$ if
  and only if $T$ is Eulerian-partitionable.
\end{lemma}

\begin{proof}
  $(\Longrightarrow)$ Suppose $D'$ is a Schnyder orientation having
  the same outdegrees as $D$.  Lemma~\ref{thm:EDGEtransform} implies
  that $T$ is partitionable into $T_0$, $T_1$, $T_2$ having the same
  homology. By Lemma~\ref{lm:homologous}, for each closed walk $W$ of
  $\pdc{G}^*$, we have $\beta(T_0,W)=\beta(T_1,W)=\beta(T_2,W)$.
  Since $D,D'$ have the same outdegrees, we have that $T$ is
  Eulerian. Consider a vertex $v$ of $\pdc{G}$ and a walk $W_v$ of
  $\pdc{G}^*$ going counterclockwise around $v$. For any oriented
  subgraph $H$ of $\pdc{G}^*$, we have
  $d_H^+(v)-d_H^-(v)=\beta(H,W_v)$, where $d_H^+(v)$ and $d_H^-(v)$ denote the outdegree and indegree of $v$ restricted to $H$, respectively. Since $T$ is Eulerian,
  we  have $\beta(T,W_v)=0$. Since
  $\beta(T_0,W_v)=\beta(T_1,W_v)=\beta(T_2,W_v)$ and $\sum
  \beta(T_i,W_v)=\beta(T,W_v)=0$, we obtain that
  $\beta(T_0,W_v)=\beta(T_1,W_v)=\beta(T_2,W_v)=0$. So each $T_i$ is
  Eulerian.

  $(\Longleftarrow)$ Suppose $T$ is Eulerian-partitionable.  Then
  Lemma~\ref{thm:EDGEtransform} implies that $D'$ is a Schnyder
  orientation. Since $T$ is Eulerian, the two orientations $D,D'$ have
  the same outdegrees.
\end{proof}

Consider $\{B_1,\ldots,B_{2g}\}$ a set of cycles of ${G}$
that forms a basis for the homology. For $\Gamma \in \mathbb{Z}^{2g}$,
an orientation of $\pdc{G}$ is of \emph{type} $\Gamma$ if
$\gamma(B_{i})=\Gamma_i$ for all $1\le i\le 2g$.

\begin{lemma}\label{th:EDGE0}
  $D'$ is a Schnyder orientation having the same outdegrees and the
  same type as $D$ (for the considered basis) if and only if $T$ is
  $0$-homologous (i.e. $D,D'$ are homologous).
\end{lemma}

\begin{proof}
  $(\Longrightarrow)$ Suppose $D'$ is a Schnyder orientation having
  the same outdegrees and the same type as $D$. Then,
  Lemma~\ref{th:EDGEeulerian} implies that $T$ is
  Eulerian-partitionable and thus Eulerian. So for any
  $F\in \hat{\mc{F^*}}$, we have $\beta(T,F)=0$. Moreover, for
  $1\le i\le 2g$, consider the region $R_i$ between $W_L(B_i)$ and
  $W_R(B_i)$ containing $B_i$. Since $T$ is Eulerian, it is going in and
  out of $R_i$ the same number of times. So
  $\beta(T,W_L(B_i)-W_R(B_i))=0$.  Since $D,D'$ have the same type, we
  have $\gamma_D(B_i)=\gamma_{D'}(B_i)$. So by
  Lemma~\ref{lem:gammaequaldelta},
  $\delta_D(W_L(B_i))+\delta_D(W_R(B_i))=\delta_{D'}(W_L(B_i))+\delta_{D'}(W_R(B_i))$. Thus
  $\beta(T,W_L(B_i)+W_R(B_i))=\beta(\Out-\Out',W_L(B_i)+W_R(B_i))=\delta_D(W_L(B_i))+\delta_D(W_R(B_i))-\delta_{D'}(W_L(B_i))-\delta_{D'}(W_R(B_i))=0$. By
  combining this with the previous equality, we obtain
  $\beta(T,W_L(B_i))=\beta(T,W_R(B_i))=0$ for all $1\le i\le 2g$. Thus
  by Lemma~\ref{lm:homologous}, we have that $T$ is 0-homologous.

  $(\Longleftarrow)$ Suppose that $T$ is $0$-homologous. Then $T$ is in
  particular Eulerian-partitionable (with the partition
  $T,\emptyset,\emptyset$). So Lemma~\ref{th:EDGEeulerian} implies
  that $D'$ is a Schnyder orientation with the same outdegrees as $D$.
  Since $T$ is $0$-homologous, by Lemma~\ref{lm:homologous}, for all
  $1\le i\le 2g$, we have
  $\beta(T,W_L(B_i))=\beta(T,W_R(B_i))=0$. Thus
  $\delta_{D}(W_L(B_i))=\beta(\Out,W_L(B_i)) =
  \beta(\Out',W_L(B_i))=\delta_{D'}(W_L(B_i))$
  and
  $\delta_{D}(W_R(B_i))=\beta(\Out,W_R(B_i)) =
  \beta(\Out',W_R(B_i))=\delta_{D'}(W_R(B_i))$.
  So by Lemma~\ref{lem:gammaequaldelta},
  $\gamma_{D}(B_i)=\delta_{D}(W_L(B_i))+\delta_{D}(W_R(B_i))=\delta_{D'}(W_L(B_i))+\delta_{D'}(W_R(B_i))=\gamma_{D'}(B_i)$.
  So $D,D'$ have the same type.
\end{proof}

Lemma~\ref{th:EDGE0} implies that when you consider Schnyder
orientations having the same outdegrees the property that they have
the same type does not depend on the choice of the basis since being
homologous does not depend on the basis. So we have the following:

  \begin{lemma}
    \label{lem:typebasis}
    If two Schnyder orientations have the same outdegrees and the same
    type (for the considered basis), then they have the same type for
    any basis.
  \end{lemma}

  Lemma~\ref{thm:EDGEtransform},~\ref{th:EDGEeulerian}
  and~\ref{th:EDGE0} are summarized in the following theorem (where by
  Lemma~\ref{lem:typebasis} we do not have to assume a particular
  choice of a basis for the third item):

\begin{theorem}\label{th:transformations}
  Let $G$ be a map on an orientable surface and $D,D'$ orientations of
  $\pdc{G}$ such that $D$ is a Schnyder orientation and $T=D\setminus
  D'$. We have the following:
  \begin{itemize}
  \item $D'$ is a Schnyder orientation if and only if $T$ is
    partitionable.
\item $D'$ is a Schnyder orientation having the same outdegrees as $D$
  if and only if $T$ is Eulerian-partitionable.
\item $D'$ is a Schnyder orientation having the same outdegrees and
  the same type as $D$ if and only if $T$ is $0$-homologous
  (i.e. $D,D'$ are homologous).
  \end{itemize}
\end{theorem}

We show in the next section that the set of Schnyder orientations that are
homologous (see third item of Theorem~\ref{th:transformations})
carries a structure of distributive lattice.

\subsection{The distributive lattice of homologous orientations}
\label{sec:lattice}

For the sake of generality, in this subsection we consider that maps
may have contractible cycles of size 1 or 2.  Consider a map $G$ on an
orientable surface and a given orientation $D_0$ of $G$. Let
$O(G,D_0)$ be the set of all the orientations of $G$ that are
homologous to $D_0$.  In this section we prove that $O(G,D_0)$ forms a
distributive lattice. We show some additional interesting properties
that are used in a recent paper by Despré, the first author, and the
third author~\cite{DGL15}.  This generalizes results for the plane
obtained by Ossona de Mendez~\cite{Oss94} and Felsner~\cite{Fel04}.
The distributive lattice structure also can also be derived from a
result of Propp~\cite{Pro93} interpreted on the dual map, see the
discussion below Theorem~\ref{th:lattice}.

In order to define an order on $O(G,D_0)$, fix an arbitrary face $f_0$
of $G$ and let $F_0$ be its counterclockwise facial walk.  Let
$\mc{F}'=\mc{F}\setminus\{F_0\}$ (where $\mc{F}$ is the set of
counterclockwise facial walks of $G$ as defined earlier). Note that
$\phi(F_0)=-\sum_{F\in\mc{F}'}\phi(F)$. Since the characteristic flows
of $\mc{F}'$ are linearly independent, any oriented subgraph of $G$
has at most one representation as a combination of characteristic
flows of $\mc{F}'$. Moreover the $0$-homologous oriented subgraphs of
$G$ are precisely the oriented subgraph that have such a
representation.  We say that a $0$-homologous oriented subgraph $T$ of
$G$ is \emph{counterclockwise} (resp. \emph{clockwise}) if its
characteristic flow can be written as a combination with positive
(resp. negative) coefficients of characteristic flows of $\mc{F}'$,
i.e. $\phi(T)=\sum_{F\in\mc{F}'}\lambda_F\phi(F)$, with
$\lambda\in\mathbb{N}^{|\mc{F}'|}$
(resp. $-\lambda\in\mathbb{N}^{|\mc{F}'|}$).  Given two orientations
$D,D'$, of $G$ we set $D\leq_{f_0} D'$ if and only if $D\setminus D'$
is counterclockwise.  Then we have the following theorem.

\begin{theorem}[\cite{Pro93}]
\label{th:lattice}
Let $G$ be a map on an orientable surface given with a particular
orientation $D_0$ and a particular face $f_0$.  Let $O(G,D_0)$ the set
of all the orientations of $G$ that are homologous to $D_0$.  We have
$(O(G,D_0),\leq_{f_0})$ is a distributive lattice.
\end{theorem}

We attribute Theorem~\ref{th:lattice} to Propp even if it is not
presented in this form in~\cite{Pro93}. Here we do not introduce Propp's formalism, but provide a new proof of Theorem~\ref{th:lattice} (as a consequence of the forthcoming Proposition~\ref{th:lattice}). This allows us to introduce notions used later in the study of this lattice. It is notable that the study of this lattice found applications in~\cite{DGL15}, where the authors found a bijection between toroidal triangulations and unicellular toroidal maps.

To prove Theorem~\ref{th:lattice}, we need to define the elementary
flips that generates the lattice.
We start by reducing the graph ${G}$. We call an edge of ${G}$
\emph{rigid with respect to $O(G,D_0)$} if it has the same orientation
in all elements of $O(G,D_0)$. Rigid edges do not play a role for the
structure of $O(G,D_0)$. We delete them from ${G}$ and call the
obtained embedded graph $\widetilde{G}$.  This graph is
embedded but it is not necessarily a map, as some faces may not be
homeomorphic to open disks. Note that if all the edges are rigid, i.e. $|O(G,D_0)|=1$, then $\widetilde{G}$
has no edges.

\begin{lemma}
\label{lem:non-rigid}
Given an edge $e$ of $G$, the following are equivalent:
\begin{enumerate}
\item $e$ is non-rigid
\item  $e$ is contained in a
$0$-homologous oriented subgraph of $D_0$
\item  $e$ is contained in a
$0$-homologous oriented subgraph of any element of $O(G,D_0)$
\end{enumerate}
\end{lemma}

\begin{proof}
  $(1 \Longrightarrow 3)$ Let $D\in O(G,D_0)$. If $e$ is non-rigid,
  then it has a different orientation in two elements $D',D''$ of
  $O(G,D_0)$.  Then we can assume by symmetry that $e$ has a different
  orientation in $D$ and $D'$ (otherwise in $D$ and $D''$ by
  symmetry). Since $D,D'$ are homologous to $D_0$, they are also
  homologous to each other. So $T=D\setminus D'$ is a $0$-homologous
  oriented subgraph of $D$ that contains $e$.

 $(3 \Longrightarrow 2)$ Trivial since $D_0\in O(G,D_0)$

  $(2 \Longrightarrow 1)$ If an edge $e$ is contained in a $0$-homologous
  oriented subgraph $T$ of $D_0$. Then let $D$ be the element of
  $O(G,D_0)$ such that $T=D_0\setminus D$. Clearly $e$ is oriented
  differently in $D$ and $D_0$, thus it is non-rigid.
\end{proof}

By Lemma~\ref{lem:non-rigid}, one can build $\widetilde{G}$ by keeping
only the edges that are contained in a $0$-homologous oriented
subgraph of $D_0$.  Note that this implies that all the edges of
$\widetilde{G}$ are incident to two distinct faces of $\widetilde{G}$.
Denote by $\widetilde{\mathcal{F}}$ the set of oriented subgraphs of
$\widetilde{G}$ corresponding to the boundaries of faces of
$\widetilde{G}$ considered counterclockwise.  Note that any
$\widetilde{F}\in \widetilde{\mathcal{F}}$ is $0$-homologous and so
its characteristic flows has a unique way to be written as a
combination of characteristic flows of $\mc{F}'$. Moreover this
combination can be written
$\phi(\widetilde{F})=\sum_{F\in X_{\widetilde{F}}}\phi(F)$, for
$X_{\widetilde{F}}\subseteq\mc{F}'$. Let $\widetilde{f}_0$ be the face
of $\widetilde{G}$ containing $f_0$ and $\widetilde{F}_0$ be the
element of $\widetilde{\mathcal{F}}$ corresponding to the boundary of
$\widetilde{f}_0$.  Let
$\widetilde{\mathcal{F}}'=\widetilde{\mathcal{F}}\setminus
\{\widetilde{F}_0\}$.
The elements of $\widetilde{\mathcal{F}}'$ are precisely the
elementary flips which suffice to generate the entire distributive
lattice $(O(G,D_0),\leq_{f_0})$.

We prove two technical lemmas concerning $\widetilde{\mathcal{F}}'$:

\begin{lemma}
\label{cl:partition}
Let $D\in O(G,D_0)$ and $T$ be a non-empty $0$-homologous oriented
subgraph of $D$.  Then there exist edge-disjoint oriented subgraphs
$T_1,\ldots,T_k$ of $D$ such that $\phi(T)=\sum_{1\leq i \leq k}
\phi(T_i)$, and, for $1\leq i \leq k$, there exists
$\widetilde{X_i}\subseteq \widetilde{{\mathcal{F}}}'$ and
$\epsilon_i\in\{-1,1\}$ such that
$\phi(T_i)=\epsilon_i\sum_{\widetilde{F}\in \widetilde{X_i}}
\phi(\widetilde{F})$.
\end{lemma}

\begin{proof}
  Since $T$ is $0$-homologous, we have
  $\phi(T)=\sum_{F\in\mathcal{F'}}\lambda_F\phi(F)$, for $\lambda\in
  \mathbb{Z}^{|\mathcal{F'}|}$. Let $\lambda_{f_0}=0$.  Thus we have
  $\phi(T)=\sum_{F\in\mathcal{F}}\lambda_F\phi(F)$.  Let
  $\lambda_{\min}=\min_{F\in\mathcal{F}}\lambda_F$ and
  $\lambda_{\max}=\max_{F\in\mathcal{F}}\lambda_F$. We may
  have $\lambda_{\min} =0$ or $\lambda_{\max} =0$ but not both since $T$
  is non-empty.  For $1\leq i\leq \lambda_{\max}$, let
  $X_{i}=\{F\in\mathcal{F'}\,|\,\lambda_F\geq i\}$ and $\epsilon_i=1$.
  Let $X_0=\emptyset$ and $\epsilon_0=1$. For $\lambda_{\min}\leq
  i\leq -1$, let $X_{i}=\{F\in\mathcal{F'}\,|\,\lambda_F\leq i\}$ and
  $\epsilon_i=-1$.  For $\lambda_{\min}\leq i \leq \lambda_{\max}$,
  let $T_i$ be the oriented subgraph such that
  $\phi(T_i)=\epsilon_i\sum_{F\in X_i}\phi(F)$.  Then we have
  $\phi(T)=\sum_{\lambda_{\min}\leq i \leq \lambda_{\max}} \phi(T_i)$.

  Since $T$ is an oriented subgraph, we have
  $\phi(T)\in\{-1,0,1\}^{|E(G)|}$. Thus for any edge of ${G}$, incident to
  faces $F_1$ and $F_2$, we have
  $(\lambda_{F_1}-\lambda_{F_2})\in\{-1,0,1\}$. So, for $1\leq i\leq
  \lambda_{\max}$, the oriented graph $T_i$ is the border between
  the faces with $\lambda$ value equal to $i$ and $i-1$.
  Symmetrically, for $\lambda_{\min}\leq i\leq -1$, the oriented graph
  $T_i$ is the border between the faces with $\lambda$ value equal
  to $i$ and $i+1$.  So all the $T_i$ are edge disjoint and are
  oriented subgraphs of $D$.

  Let $\widetilde{X_i}=\{\widetilde{F}\in\widetilde{\mathcal{F}}'\,|
  \,\phi(\widetilde{F})=\sum_{F\in X'} \phi(F) \textrm{ for some }
  X'\subseteq X_i\}$.  Since $T_i$ is $0$-homologous, the edges of
  $T_i$ can be reversed in $D$ to obtain another element of
  $O(G,D_0)$. Thus there is no rigid edge in $T_i$.  Thus
  $\phi(T_i)=\epsilon_i\sum_{F\in
    X_i}\phi(F)=\epsilon_i\sum_{\widetilde{F}\in\widetilde{X_i}}\phi(\widetilde{F})$.
\end{proof}

\begin{lemma}
\label{cl:facenonrigid}
Let $D\in O(G,D_0)$ and $T$ be a non-empty $0$-homologous oriented
subgraph of $D$ such that there exists $\widetilde{X}\subseteq
\widetilde{{\mathcal{F}}}'$ and $\epsilon\in\{-1,1\}$ satisfying
$\phi(T)=\epsilon\sum_{\widetilde{F}\in \widetilde{X}}
\phi(\widetilde{F})$. Then there exists $\widetilde{F}\in\widetilde{X}$ such
that $\epsilon\,\phi(\widetilde{F})$ corresponds to an oriented
subgraph of $D$.
\end{lemma}

\begin{proof}
  The proof is done by induction on $|\widetilde{X}|$.  Assume that
  $\epsilon =1$ (the case $\epsilon =-1$ is proved similarly).

  If $|\widetilde{X}|=1$, then the conclusion is clear since
  $\phi(T)=\sum_{\widetilde{F}\in\widetilde{X}}\phi(\widetilde{F})$.
  We now assume that $|\widetilde{X}|> 1$. Towards a contradiction, suppose
  that for any $\widetilde{F}\in\widetilde{X}$ we do not have the
  conclusion, i.e $\phi(\widetilde{F})_e\neq\phi(T)_e$ for some $e \in
  \widetilde{F}$. Let $\widetilde{F_1}\in\widetilde{X}$ and $e \in
  \widetilde{F_1}$ such that $\phi(\widetilde{F_1})_e\neq\phi(T)_e$.
  Since $\widetilde{F_1}$ is counterclockwise, we have
  $\widetilde{F_1}$ on the left of $e$.  Let
  $\widetilde{F_2}\in\widetilde{{\mathcal{F}}} $ that is on the right
  of $e$. Note that $\phi(\widetilde{F_1})_e=-\phi(\widetilde{F_2})_e$
  and for any other face $\widetilde{F}\in \widetilde{\mathcal{F}}$,
  we have $\phi(\widetilde{F})_e=0$.  Since
  $\phi(T)=\sum_{\widetilde{F}\in\widetilde{X}}\phi(\widetilde{F})$,
  we have $\widetilde{F_2}\in\widetilde{X}$ and $\phi(T)_e=0$. By possibly 
  swapping the role of $\widetilde{F_1}$ and $\widetilde{F_2}$, 
  we can assume that $\phi(D)_e=\phi(\widetilde{F_1})_e$, i.e.,
  $e$ is oriented the same way in $\widetilde{F_1}$ and $D$. 
  Since $e$ is not rigid, there exists an orientation $D'$ in
  $O(G,D_0)$ such that $\phi(D)_e=-\phi(D')_e$.

  Let $T'$ be the non-empty $0$-homologous oriented subgraph of $D$
  such that $T'=D\setminus D'$. Lemma~\ref{cl:partition} implies that
  there exists edge-disjoint oriented subgraphs $T_1,\ldots,T_k$ of
  $D$ such that $\phi(T)=\sum_{1\leq i \leq k} \phi(T_i)$, and, for
  $1\leq i \leq k$, there exists $\widetilde{X_i}\subseteq
  \widetilde{{\mathcal{F}}}'$ and $\epsilon_i\in\{-1,1\}$ such that
  $\phi(T_i)=\epsilon_i\sum_{\widetilde{F}\in \widetilde{X_i}}
  \phi(\widetilde{F})$.  Since $T'$ is the disjoint union of
  $T_1,\ldots,T_k$, there exists $1\leq i\leq k$, such that $e$ is an
  edge of $T_i$. Assume by symmetry that $e$ is an edge of
  $T_1$. Since $\phi(T_1)_e=\phi(D)_e=\phi(\widetilde{F_1})_e$, we
  have $\epsilon_1=1$, $\widetilde{F_1}\in \widetilde{X_1}$ and
  $\widetilde{F_2}\notin \widetilde{X_1}$.

  Let $\widetilde{Y}=\widetilde{X}\cap \widetilde{X_1}$. Thus
  $\widetilde{F_1}\in \widetilde{Y}$ and $\widetilde{F_2}\notin
  \widetilde{Y}$. So $|\widetilde{Y}|<|\widetilde{X}|$. Let
  $T_{\widetilde{Y}}$ be the oriented subgraph of $G$ such that
  $T_{\widetilde{Y}}=\sum_{\widetilde{F}\in \widetilde{Y}}
  \phi(\widetilde{F})$.  Note that the edges of $T$ (resp. $T_1$) are
  those incident to exactly one face of $\widetilde{X}$
  (resp. $\widetilde{X_1}$). Similarly every edge of
  $T_{\widetilde{Y}}$ is incident to exactly one face of
  $\widetilde{Y}=\widetilde{X}\cap \widetilde{X_1}$, i.e. it has one
  incident face in $\widetilde{Y}=\widetilde{X}\cap \widetilde{X_1}$
  and the other incident face not in $\widetilde{X}$ or not in
  $\widetilde{X_1}$.  In the first case this edge is in $T$, otherwise
  it is in $T_1$. So every edge of $T_{\widetilde{Y}}$ is an edge of
  $T\cup T_1$. Hence $T_{\widetilde{Y}}$ is an oriented subgraph of
  $D$. So we can apply the induction hypothesis on
  $T_{\widetilde{Y}}$. This implies that there exists
  $\widetilde{F}\in\widetilde{Y}$ such that $\widetilde{F}$ is an
  oriented subgraph of $D$. Since
  $\widetilde{Y}\subseteq\widetilde{X}$, this is a contradiction to
  our assumption.
\end{proof}

We need the following
characterization of distributive lattice from~\cite{Fel09}:

\begin{theorem}[\cite{Fel09}]
\label{th:hasse}
  An oriented graph $\mathcal{H}=(V,E)$ is the Hasse diagram of a
  distributive lattice if and only if it is connected, acyclic, and
  admits an edge-labeling $c$ of the edges such that:
\begin{itemize}
\item if $(u,v), (u,w)\in E$ then 
    \begin{itemize}
      \item[(U1)] $c(u,v)\neq c(u,w)$ and
      \item[(U2)] there is $z\in V$ such that $(v,z), (w,z)\in E$,
        $c(u,v)=c(w,z)$, and $c(u,w)=c(v,z)$.
    \end{itemize}
\item if $(v,z), (w,z)\in E$ then
    \begin{itemize}
      \item[(L1)] $c(v,z)\neq c(w,z)$ and
      \item[(L2)] there is $u\in V$ such that $(u,v), (u,w)\in E$,
        $c(u,v)=c(w,z)$, and $c(u,w)=c(v,z)$.
    \end{itemize}
\end{itemize}
\end{theorem}

We define the directed graph $\mathcal{H}$ with vertex set
$O(G,D_0)$.  There is an oriented edge from $D_1$ to $D_2$ in
$\mathcal{H}$ (with $D_1\leq_{f_0}D_2$) if and only if
$D_1\setminus D_2\in \widetilde{\mathcal{F}}'$.  We define the label
of that edge as $c(D_1,D_2)=D_1\setminus D_2$. We show that
$\mathcal{H}$ fulfills all the conditions of Theorem~\ref{th:hasse},
and thus obtain the following:

\begin{proposition}
  \label{lem:fulfillshasse}
  $\mathcal{H}$ is the Hasse diagram of a distributive lattice.
\end{proposition}

\begin{proof}
  The characteristic flows of elements of $\widetilde{\mathcal{F}}'$
  form an independent set, hence the digraph $\mathcal{H}$ is acyclic.
  By definition all outgoing and all incoming edges of a vertex of
  $\mathcal{H}$ have different labels, i.e. the labeling $c$ satisfies
  (U1) and (L1).  If $(D_u,D_v)$ and $(D_u,D_w)$ belong to
  $\mathcal{H}$, then $T_v=D_u\setminus D_v$ and
  $T_w=D_u\setminus D_w$ are both elements of
  $\widetilde{\mathcal{F}}'$, so they must be edge disjoint.  Thus,
  the orientation $D_z$ obtained from reversing the edges of $T_w$ in
  $D_v$ or equivalently $T_v$ in $D_w$ is in $O(G,D_0)$. This gives
  (U2). The same reasoning gives (L2).  It remains to show that
  $\mathcal{H}$ is connected.

Given a $0$-homologous oriented subgraph $T$ of $G$,
such that $T=\sum_{F\in\mc{F}'}\lambda_F\phi(F)$, we define
$s(T)=\sum_{F\in\mathcal{F'}}|\lambda_F|$.

Let $D,D'$ be two orientations in $O(G,D_0)$, and $T=D\setminus
D'$. We prove by induction on $s(T)$ that $D,D'$ are connected in
$\mathcal{H}$.  This is clear if $s(T)=0$ as then $D=D'$. So we now
assume that $s(T)\neq 0$ and so that $D,D'$ are distinct.
Lemma~\ref{cl:partition} implies that there exists edge-disjoint
oriented subgraphs $T_1,\ldots,T_k$ of $D$ such that
$\phi(T)=\sum_{1\leq i \leq k} \phi(T_i)$, and, for $1\leq i \leq k$,
there exists $\widetilde{X_i}\subseteq \widetilde{{\mathcal{F}}}'$ and
$\epsilon_i\in\{-1,1\}$ such that
$\phi(T_i)=\epsilon_i\sum_{\widetilde{F}\in \widetilde{X_i}}
\phi(\widetilde{F})$.  Lemma~\ref{cl:facenonrigid} applied to $T_1$
implies that there exists $\widetilde{F_1}\in\widetilde{X_1}$ such
that $\epsilon_1\,\phi(\widetilde{F_1})$ corresponds to an oriented subgraph
of $D$.  Let $T'$ be the oriented subgraph such that
$\phi(T)=\epsilon_1\phi(\widetilde{F_1})+\phi(T')$. Thus:
\begin{eqnarray*}
\phi(T') &=&\phi(T)-\epsilon_1\phi(\widetilde{F_1})\\
&=&\sum_{1\leq i\leq
  k}\phi(T_i)-\epsilon_1\phi(\widetilde{F_1})\\
&=&\sum_{\widetilde{F}\in(\widetilde{X_1}\setminus\{\widetilde{F_1}\})}
\epsilon_1 \phi(\widetilde{F})+\sum_{2\leq i\leq k}\sum_{\widetilde{F}\in
  \widetilde{X_i}} \epsilon_i\phi(\widetilde{F})\\
&=&\sum_{\widetilde{F}\in(\widetilde{X_1}\setminus\{\widetilde{F_1}\})}
\sum_{F\in{X_{\widetilde{F}}}} \epsilon_1\phi({F})+\sum_{2\leq i\leq
  k}\sum_{\widetilde{F}\in
  \widetilde{X_i}}\sum_{F\in{X_{\widetilde{F}}}} \epsilon_i\phi({F})
\end{eqnarray*}
So $T'$ is $0$-homologous.  Let $D''$ be such that $\epsilon_1
\widetilde{F_1}=D\setminus D''$.  So we have $D'' \in O(G,D_0)$ and
there is an edge between $D$ and $D''$ in $\mathcal{H}$. Moreover
$T'=D''\setminus D'$ and $s(T')=
s(T)-|{X_{\widetilde{F_1}}}|<s(T)$. So the induction hypothesis on
$D'',D'$ implies that they are connected in $\mathcal{H}$. So $D,D'$
are also connected in $\mathcal{H}$.
\end{proof}

Note that Proposition~\ref{lem:fulfillshasse} gives a  proof of
Theorem~\ref{th:lattice} independent from
Propp~\cite{Pro93}. 

We continue to further investigate the set
$O(G,D_0)$.

\begin{proposition}
\label{lem:necessary}
For every element $ \widetilde{F}\in \widetilde{\mathcal{F}}$, there
exists $D$ in $O(G,D_0)$ such that $\widetilde{F}$ is an oriented subgraph
of $D$. 
\end{proposition}

\begin{proof}
  Let $ \widetilde{F}\in \widetilde{\mathcal{F}}$. Let $D$ be an
  element of $O(G,D_0)$ that maximizes the number of edges of
  $\widetilde{F}$ that have the same orientation in $\widetilde{F}$
  and $D$, i.e. $D$ maximizes the number of edges oriented
  counterclockwise on the boundary of the face of $\widetilde{G}$
  corresponding to $\widetilde{F}$. Towards a contradiction, suppose that
  there is an edge $e$ of $\widetilde{F}$ that does not have the same
  orientation in $\widetilde{F}$ and $D$. The edge $e$ is in
  $\widetilde{G}$ so it is non-rigid.  Let $D'\in O(G,D_0)$ such that
  $e$ is oriented differently in $D$ and $D'$. Let $T=D\setminus D'$.
  By Lemma~\ref{cl:partition}, there exist edge-disjoint oriented
  subgraphs $T_1,\ldots,T_k$ of $D$ such that
  $\phi(T)=\sum_{1\leq i \leq k} \phi(T_i)$, and, for
  $1\leq i \leq k$, there exists
  $\widetilde{X}_i\subseteq \widetilde{{\mathcal{F}}}'$ and
  $\epsilon_i\in\{-1,1\}$ such that
  $\phi(T_i)=\epsilon_i\sum_{\widetilde{F}'\in \widetilde{X}_i}
  \phi(\widetilde{F}')$.
  W.l.o.g., we can assume that $e$ is an edge of $T_1$.  Let $D''$ be
  the element of $O(G,D_0)$ such that $T_1=D\setminus D''$.  The
  oriented subgraph $T_1$ intersects $\widetilde{F}$ only on edges of
  $D$ oriented clockwise on the border of $\widetilde{F}$. So $D''$
  contains strictly more edges oriented counterclockwise on the border
  of the face $\widetilde{F}$ than $D$, a contradiction.  So all the
  edges of $\widetilde{F}$ have the same orientation in $D$.  So
  $\widetilde{F}$ is a $0$-homologous oriented subgraph of $D$.
\end{proof}

By Proposition~\ref{lem:necessary}, for every element
$ \widetilde{F}\in \widetilde{\mathcal{F}}'$ there exists $D$ in
$O(G,D_0)$ such that $\widetilde{F}$ is an oriented subgraph of $D$. Thus
there exists $D'$ such that $\widetilde{F}=D\setminus D'$ and $D,D'$
are linked in $\mathcal{H}$.  Thus, $\widetilde{\mathcal{F}}'$ is a minimal 
set that generates the lattice.

A distributive lattice has a unique maximal (resp. minimal) element.
Let $D_{\max}$ (resp. $D_{\min}$) be the maximal (resp. minimal)
element of $(O(G,D_0),\leq_{f_0})$.

\begin{proposition}
\label{lem:maxtilde}
  $\widetilde{F}_0$ (resp. $-\widetilde{F}_0$) is an oriented subgraph
  of $D_{\max}$ (resp. $D_{\min}$). 
\end{proposition}

\begin{proof}
  By Proposition~\ref{lem:necessary}, there exists $D$ in $O(G,D_0)$ such
  that $\widetilde{F_0}$ is an oriented subgraph of $D$. Let
  $T=D\setminus D_{\max}$. Since $D\leq_{f_0} D_{\max}$, the
  characteristic flow of $T$ can be written as a combination with
  positive coefficients of characteristic flows of
  $\widetilde{\mathcal{F}}'$, i.e.
  $\phi(T)=\sum_{\widetilde{F}\in
    \widetilde{\mathcal{F}}'}\lambda_F\phi(\widetilde{F})$
  with $\lambda\in\mathbb{N}^{|\mc{F}'|}$. So $T$ is disjoint from
  $\widetilde{F}_0$.  Thus $\widetilde{F}_0$ is an oriented subgraph
  of $D_{\max}$. The proof is analogous for $D_{\min}$.
  \end{proof}

\begin{proposition}
\label{prop:maximal}
$D_{\max}$ (resp. $D_{\min}$) contains no counterclockwise
(resp. clockwise) non-empty $0$-homologous oriented subgraph.
\end{proposition}

\begin{proof}
  Towards a contradiction, suppose that $D_{\max}$ contains a
  counterclockwise non-empty $0$-homologous oriented subgraph
  $T$. Then there exists $D\in O(G,D_0)$ distinct from $D_{\max}$ such
  that $T=D_{\max}\setminus D$. We have $D_{\max}\leq_{f_0} D$ by
  definition of $\leq_{f_0}$, a contradiction to the maximality of
  $D_{\max}$.
\end{proof}

In the definition of counterclockwise (resp. clockwise)
non-empty $0$-homologous oriented subgraph, used in
Proposition~\ref{prop:maximal}, the sum is taken over elements of
$\mc{F}'$ and thus does not use $F_0$. In particular, $D_{\max}$
(resp. $D_{\min}$) may contain regions whose boundary is oriented
counterclockwise (resp. clockwise) according to the region but then
such a region contains $F_0$.

We conclude this section by applying Theorem~\ref{th:lattice} to
Schnyder orientations:

\begin{theorem}
\label{cor:lattice}
  Let $G$ be a map on an orientable surface given with a particular
  Schnyder orientation $D_0$ of $\pdc{G}$ and a particular face $f_0$ of
  $\pdc{G}$.  Let $S(\pdc{G},D_0)$ be the set of all the Schnyder
  orientations of $\pdc{G}$ that have the same outdegrees and same
  type as $D_0$.  We have that $(S(\pdc{G},D_0),\leq_{f_0})$ is a
  distributive lattice.
\end{theorem}

\begin{proof}
  By the third item of Theorem~\ref{th:transformations}, we have
  $S(\pdc{G},D_0)=O(\pdc{G},D_0)$. Then the conclusion holds by
  Theorem~\ref{th:lattice}.
\end{proof}

Theorem~\ref{cor:lattice} is illustrated in Section~\ref{sec:example}
on an example. Note that the minimal element of the lattice and its
properties (Proposition~\ref{lem:fulfillshasse} to~\ref{prop:maximal})
are used in~\cite{DGL15} to obtain a new bijection concerning toroidal
triangulations.

\section{Toroidal triangulations}
\label{sec:toretri}

\subsection{New proof of the existence of Schnyder woods}
\label{sec:proof}

In this section we look specifically at the case of toroidal
triangulations. We study the structure of 3-orientations of toroidal
triangulations and show how one can use it to prove the existence of
Schnyder woods in toroidal triangulations. This corresponds to the case
$g=1$ of Conjecture~\ref{conjecture}. Given a toroidal triangulation
$G$, a \emph{3-orientation} of $G$ is an orientation of the edges of
$G$ such that every vertex has outdegree exactly three. By
Theorem~\ref{th:barat}, a simple toroidal triangulation admits a
3-orientation. This can be shown to be true also for 
non-simple triangulations, for example using edge-contraction.

Consider a toroidal triangulation $G$ and a 3-orientation of $G$.  Let
$G^\infty$ be the universal cover of $G$.

\begin{lemma}
\label{lem:kmoins3}
A cycle $C$ of $G^\infty$ of length $k$ has exactly $k-3$ edges
leaving $C$ and directed towards the interior of $C$.
\end{lemma}

\begin{proof}
  Let $x$ be the number of edges leaving $C$ and directed towards the
  interior of $C$.  Consider the cycle $C$ and its interior as a planar
  graph $C^{o}$. Euler's formula gives $n-m+f=2$ where $n,m,f$ are
  respectively the number of vertices, edges and faces of $C^{o}$.
  Every inner vertex has exactly outdegree three, so
  $m=3(n-k)+k+x$. Every inner face is a triangle so $2m=3(f-1)+k$.
  The last two equalities can be used to replace $f$ and $m$ in
  Euler's formula, and  obtain $x=k-3$.
\end{proof}

 For an edge $e$ of $G$, we define
the \emph{middle walk from $e$} as the sequence of edges $(e_i)_{i\geq
  0}$ obtained by the following method.  Let $e_0=e$. If the edge
$e_i$ is entering a vertex $v$, then the edge $e_{i+1}$ is chosen in
the three edges leaving $v$ as the edge in the ``middle'' coming from
$e_i$ (i.e. $v$ should have exactly one edge leaving on the left of
the path consisting of the two edges $e_i,e_{i+1}$ and thus exactly
one edge leaving on the right).

A directed cycle $M$ of $G$ is said to be a \emph{middle cycle} if
every vertex $v$ of $M$ has exactly one edge leaving $v$ on the left
of $M$ (and thus exactly one edge leaving $v$ on the right of $M$).
Note that if $M$ is a middle cycle, and $e$ is an edge of $M$, then
the middle walk from $e$ consists of the sequence of edges of $M$
repeated periodically.  Note that a middle cycle is not contractible,
otherwise in $G^\infty$ it forms a contradiction to
Lemma~\ref{lem:kmoins3}. Similar arguments lead to:

   \begin{lemma}
\label{lem:middleequal}
Two middle cycles that are  weakly homologous are either
vertex-disjoint or equal.
   \end{lemma}

We have the following useful lemma concerning
middle walks and middle cycles:

\begin{lemma}
\label{lem:middlecycle}
A middle walk always ends on a middle cycle.
\end{lemma}

\begin{proof}
  Start from any edge $e_0$ of $G$ and consider the middle walk
  $W=(e_i)_{i\geq 0}$ from $e_0$.  The graph $G$ has a finite number
  of edges, so some edges will be used several times in $W$.  Consider a
  minimal subsequence $e_k, \ldots, e_\ell$ such that no edge appears
  twice and $e_k=e_{\ell+1}$.  Thus $W$ ends periodically on the
  sequence of edges $e_k, \ldots, e_\ell$. We prove that $e_k, \ldots,
  e_\ell$ is a middle cycle.

  Assume that $k=0$ for simplicity. Thus $e_0, \ldots, e_\ell$ is an
  Eulerian subgraph $E$. If $E$ is a cycle then it is a middle cycle
  and we are done. So we can consider that it visits some vertices
  several times.  Let $e_i, e_j$, with $0\leq i < j\leq \ell$, such
  that $e_i, e_j$ are both leaving the same vertex $v$. By definition
  of $\ell$, we have $e_i\neq e_j$. Let $A$ and
  $B$ be the two closed walks  $e_i,\ldots, e_{j-1}$ and $e_{j},\ldots, e_{i-1}$,
  respectively, where indices are modulo $\ell+1$.

  Consider a copy $v_0$ of $v$ in the universal cover $G^\infty$. 
  Define the walk $P$ obtained by starting at $v_0$ following the
  edges of $G^\infty$ corresponding to the edges of $A$, and then to the
  edges of $B$. Similarly, define the walk $Q$ obtained by starting at
  $v_0$ following the edges of $B$, and then the edges of $A$.  The two
  walks $P$ and $Q$ both start at $v_0$ and both end at the same
  vertex $v_1$ that is a copy of $v$. Note that $v_1$ and $v_0$
  may coincide. All the vertices that
  are visited on the interior of $P$ and $Q$ have exactly one edge
  leaving on the left and exactly one edge leaving on the right.  The
  two walks $P$ and $Q$ may intersect before they end at $v_1$ thus we
  define $P'$ and $Q'$ has the subwalks of $P$ and $Q$ starting at
  $v_0$, ending on the same vertex $u$ (possibly distinct from $v_1$ or
  not) and such that $P'$ and $Q'$ are not intersecting on their
  interior vertices. Then the union of $P'$ and $Q'$ forms a cycle $C$
  of $G^\infty$. All the vertices of $C$ except possibly $v_0$ and $u$,
  have exactly one edge leaving $C$ and directed towards the interior
  of $C$, a contradiction to Lemma~\ref{lem:kmoins3}.
\end{proof}

A consequence of Lemma~\ref{lem:middlecycle} is that any 3-orientation
of a toroidal triangulation has a middle cycle.  The 3-orientation of
the toroidal triangulation on the left of Figure~\ref{fig:orientation}
is an example where there is a unique middle cycle (the diagonal).  We
show in Lemma~\ref{lem:2middlecycle} that for any toroidal
triangulation there exists a 3-orientation with several middle cycles.

Note that a middle cycle $C$ satisfies $\gamma(C)=0$ (when $C$ is
considered in any direction). So, by Lemma~\ref{lem:middlecycle},
there is always a cycle with value $\gamma$ equal to $0$ in a
3-orientation of a toroidal triangulation.

The orientation of the toroidal triangulation on the left of
Figure~\ref{fig:orientation} is an example of a 3-orientation of a
toroidal triangulation where some cycles have value $\gamma$ not equal
to $0$.  The value of $\gamma$ for the three loops is $2, 0$ and $-2$.

Two non-contractible not weakly homologous cycles generate the
homology of the torus with respect to $\mathbb{Q}$. That is if
$B_1,B_2$ are non contractible cycles that are  not weakly
homologous, then for any cycle $C$ there exists
$k,k_1,k_2 \in \mathbb Z$, $k\neq 0$, such that $kC$ is homologous to
$k_1 B_1 + k B_{2}$.

\begin{lemma}
  \label{lem:deltagammatri} In a $3$-orientation, 
consider $B_1,B_2,C$ are non contractible cycles, such that
$B_1,B_2$ are not weakly homologous. Let
$k,k_1,k_2 \in \mathbb Z$, $k\neq 0$ such that $kC$ is homologous to
$k_1 B_1 + k_2B_{2}$. Then
$k\gamma(C)=k_1\,\gamma (B_1) + k_2\,\gamma (B_2)$.
\end{lemma}

\begin{proof} 
  Let $v$ be a vertex in the intersection of $B_1$ and $B_2$. Consider
  a drawing of $G^\infty$ obtained by replicating a flat
  representation of $G$ to tile the plane.  Let $v_0$ be a copy of
  $v$. Consider the path $B$ starting at $v_0$ and following $k_1$
  times the edges corresponding to $B_1$ and then $k_2$ times the
  edges corresponding to $B_2$ (we are going backwards if $k_i$ is
  negative). This path ends at a copy $v_1$ of $v$.  Since $C$ is
  non-contractible we have $k_1$ or $k_2$ not equal to $0$ and thus
  $v_1$ is distinct from $v_0$. Let $B^\infty$ be the infinite path
  obtained by replicating $B$ (forwards and backwards) from
  $v_0$. Since $kC$ is homologous to $k_1 B_1 + k_{2}B_{2}$ we can
  find an infinite path $C^\infty$, that corresponds to copies of $C$
  replicated, that does not intersect $B^\infty$ and situated on the
  right side of $B$. Now we can find a
  copy $B'^\infty$ of $B^\infty$, such that $C^\infty$ lies between
  $B^\infty$ and $B'^\infty$ without intersecting them. Choose a copy
  $v'_0$ of $v$ on $B'^\infty$. Let $B'$ be the copy of $B$ starting
  at $v'_0$ and ending at a vertex $v'_1$.  Let $R$ be the region
  bounded by $B,B'$ and the segments $[v_0,v'_0], [v_1,v_1']$.

  Consider the toroidal triangulation $H$ whose representation is $R$
  (obtained by identifying $B,B'$ and $[v_0,v'_0],[v_1,v_1']$). Note
  that $H$ is just made of several copies of $G$.  Let $C'$ be the
  subpath of $C^\infty$ intersecting the region $R$ corresponding to
  exactly one copy of $kC$.  Let $R_1$ be the subregion of $R$ bounded
  by $B$ and $C'$ and $R_2$ the subregion of $R$ bounded by $B'$ and
  $C'$. By some counting arguments (Euler's formula + triangulation +
  3-orientation) in the region $R_1$ and $R_2$, we obtain that
  $\gamma(C')=\gamma(B)$ and thus
  $k\gamma(C)=k_1 \gamma (B_1) + k_{2} \gamma (B_{2})$.
\end{proof}

By Lemma~\ref{lem:middlecycle}, a middle walk $W$ always ends on a
middle cycle. Let us denote by $M_W$ this middle cycle and $P_W$ the
part of $W$ before $M_W$. Note that $P_W$ may be empty. We say that a
middle walk is leaving a cycle $C$ if its starting edge is incident to
$C$ and leaving $C$.

Let us now prove the main lemma of this section.

\begin{lemma}
\label{lem:2middlecycle}
$G$ admits a 3-orientation with two middle cycles that are not
weakly homologous. 
\end{lemma}

\begin{proof}
  Towards a contradiction, suppose that there is no 3-orientation of
  $G$ with two middle cycles that are not weakly homologous. We first
  prove the following claim:

  \begin{claim}
\label{cl:middleequal}
There exists a 3-orientation of $G$ with a middle cycle $M$, a middle
walk $W$ leaving $M$ and $M_W=M$.
  \end{claim}

 \begin{proofclaim}
   Towards a contradiction, suppose that there is no 3-orientation of $G$ with
   a middle cycle $M$, a middle walk $W$ leaving $M$ and $M_W=M$. We
   first prove the following:

\begin{sclaim}
\label{claim:middleinterior}
Any 3-orientation of $G$, middle cycle $M$ and middle walk $W$ leaving
$M$ are such that $M$ does not intersect the interior of $W$.
\end{sclaim}

\begin{proofsclaim}
  Towards a contradiction, suppose that $M$ intersects the interior of
  $W$.  By assumption, cycles $M_W$ and $M$ are weakly homologous and
  $M_W\neq M$. Thus by Lemma~\ref{lem:middleequal}, they are
  vertex-disjoint.  So $M$ intersects the interior of $P_W$.  Assume
  by symmetry that $P_W$ is leaving $M$ on its left side.  If $P_W$ is
  entering $M$ from its left side, in $G^\infty$, the edges of $P_W$
  plus $M$ form a cycle contradicting Lemma~\ref{lem:kmoins3}. So
  $P_W$ is entering $M$ from its right side. Hence $M_W$ intersects
  the interior of $P_W$ on a vertex $v$. Let $e$ be the edge of $P_W$
  leaving $v$. Then the middle cycle $M_W$ and the middle walk $W'$
  started on $e$ satisfies $M_{W'}=M_W$, contradicting the
  hypothesis. So $M$ does not intersect the interior of $W$.
\end{proofsclaim}

Consider a 3-orientation, a middle cycle $M$ and a middle walk $W$
leaving $M$ such that the length of $P_W$ is maximized.  By assumption
$M_{W}$ is weakly homologous to $M$. Assume by symmetry that $P_W$ is
leaving $M$ on its left side. By assumption $M_W\neq M$.
(\ref{claim:middleinterior}) implies that $M$ does not intersect the
interior of $W$. Let $v$ (resp. $e_0$) be the starting vertex
(resp. edge) of $W$.  Consider now the 3-orientation obtained by
reversing $M_W$.  Consider the middle walk $W'$ started at $e_0$.
Walk $W'$ follows $P_W$, then arrives on $M_W$ and crosses it (since
$M_W$ has been reversed).  (\ref{claim:middleinterior}) implies that
$M$ does not intersect the interior of $W'$.  Similarly,
(\ref{claim:middleinterior}) applied to $M_W$ and $W'\setminus P_W$
(the walk obtained from $W'$ by removing the first edges corresponding
to $P_W$), implies that $M_W$ does not intersect the interior of
$W'\setminus P_W$. Thus, $M_{W'}$ is weakly homologous to $M_W$
and $M_{W'}$ is in the interior of the region between $M$ and $M_W$
on the right of $M$. Thus $P_{W'}$ strictly contains $P_W$
and is thus longer, a contradiction.
  \end{proofclaim}

  By Claim~\ref{cl:middleequal}, consider a 3-orientation of $G$ with
  a middle cycle $M$ and a middle walk $W$ leaving $M$ such that
  $M_W=M$. Note that $W$ is leaving $M$ from one side and entering it
  in the other side, otherwise $W$ and $M$ contradicts
  Lemma~\ref{lem:kmoins3}.  Let $e_0$ be the starting edge of $W$. Let
  $v,u$ be the starting and ending point of $P_W$, respectively, where
  $u=v$ may occur.  Consider the 3-orientation obtained by reversing
  $M$. Let $Q$ be the directed path from $u$ to $v$ along $M$ ($Q$ is
  empty if $u=v$). Let $C$ be the directed cycle $P_W\cup Q$. We
  compute the value $\gamma$ of $C$. If $u\neq v$, then $C$ is almost
  everywhere a middle cycle, except at $u$ and $v$. At $u$, it has two
  edges leaving on its right side, and at $v$ it has two edges leaving
  on its left side. So we have $\gamma(C)=0$. If $u=v$, then $C$ is a
  middle cycle and $\gamma(C)=0$.  Thus, in any case
  $\gamma(C)=0$. Note that furthermore $\gamma(M)=0$ holds.  The two
  cycles $M,C$ are non contractible and not weakly homologous
so any non-contractible
  cycle of $G$ has $\gamma$ equal to zero by
  Lemma~\ref{lem:deltagammatri}.

  Consider the middle walk $W'$ from $e_0$. By assumption $M_{W'}$ is
  weakly homologous to $M$.  The beginning $P_{W'}$ is the same as for
  $P_W$. As we have reversed the edges of $M$, when arriving on $u$,
  path $P_{W'}$ crosses $M$ and continues until reaching
  $M_{W'}$. Thus $M_{W'}$ intersects the interior of $P_{W'}$ at a
  vertex $v'$. Let $u'$ be the ending point of $P_{W'}$ (note that we
  may have $u'=v'$). Let $P'$ be the non-empty subpath of $P_{W'}$
  from $v'$ to $u'$. Let $Q'$ be the directed path from $u'$ to $v'$
  along $M_{W'}$ ($Q'$ is empty if $u'=v'$). Let $C'$ be the
  non-contractible directed cycle $P'\cup Q'$. We compute
  $\gamma(C')$. The cycle $C'$ is almost everywhere a middle cycle,
  except at $v'$. At $v'$, it has two edges leaving on its left or
  right side, depending on $M_{W'}$ crossing $P_{W'}$ from its left or
  right side. Thus, we have $\gamma(C')=\pm 2$, a contradiction.
\end{proof}

By Lemma~\ref{lem:2middlecycle}, for any toroidal triangulation, there
exists a 3-orientation with two middle cycles that are not weakly
homologous. By Lemma~\ref{lem:deltagammatri}, any non-contractible
cycle of $G$ has value $\gamma$ equal to zero. Note that $\gamma(C)=0$
for any non-contractible cycle $C$ does not necessarily imply the
existence of two middle cycle that are not weakly homologous.  The
3-orientation of the toroidal triangulation of Figure~\ref{fig:gamma0}
is an example where $\gamma(C)=0$ for any non-contractible cycle $C$
but all the middle cycle are weakly homologous. The colors should help
the reader to compute all the middle cycles by starting from any edge
and following the colors. One can see that all the middle cycles are
vertical (up or down) and that the horizontal (non-directed) cycle has
value $\gamma$ equal to $0$ so we have $\gamma$ equal to $0$
everywhere.  Of course, the colors also show the underlying Schnyder
wood.

\begin{figure}[!h]
\center
\includegraphics[scale=0.5]{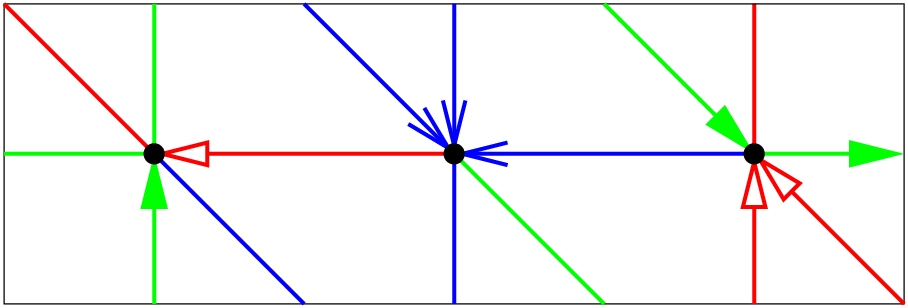} 
\caption{A 3-orientation of a toroidal triangulation with
  $\gamma(C)=0$ for any non-contractible cycle $C$. All the middle
  cycle are weakly homologous.}
\label{fig:gamma0}
\end{figure}

By combining Lemma~\ref{lem:2middlecycle} and
Theorem~\ref{th:characterizationgamma}, we obtain the following:

\begin{theorem}
\label{th:triagnulationtore}
A toroidal triangulation admits a 1-{\sc edge}, 1-{\sc vertex}, 1-{\sc
  face} angle labeling and thus a  Schnyder wood.
\end{theorem}

\begin{proof} 
  By Lemma~\ref{lem:2middlecycle}, there exists a 3-orientation with
  two middle cycles that are not weakly homologous. By
  Lemma~\ref{lem:deltagammatri}, any non-contractible cycle of $G$ has
  value $\gamma$ equal to zero.  Thus by
  Theorem~\ref{th:characterizationgamma}, this implies that the
  orientation corresponds to an {\sc edge} angle labeling.  Then by
  Lemma~\ref{lem:EDGElabeling}, the labeling is also {\sc vertex} and
  {\sc face}.  As all the edges are oriented in one direction only, it
  is 1-{\sc edge}.  As all the vertices have outdegree three, it is
  1-{\sc vertex}. Finally as all the faces are triangles it is 1-{\sc
    face} (in the corresponding orientation of $\pdc{G}$, all the
  edges incident to dual-vertices are outgoing). By
  Proposition~\ref{prop:bijtore}, this 1-{\sc edge}, 1-{\sc vertex},
  1-{\sc face} angle labeling corresponds to a Schnyder wood.
\end{proof}

Theorem~\ref{th:triagnulationtore} corresponds to the case $g=1$ of
Conjecture~\ref{conjecture}.  By~\cite{GL13}, we already knew that
Schnyder woods exist for toroidal triangulations, but this section provides
an alternative proof based on the structure of
3-orientations and the characterization theorem of
Section~\ref{sec:char}.

\subsection{The crossing property}
\label{sec:crossing}

A  Schnyder wood of a toroidal triangulation is
\emph{crossing}, if for each pair $i,j$ of different colors, there
exist a monochromatic cycle of color $i$ intersecting a monochromatic
cycle of color $j$.  In~\cite{GL13} a strengthening of Theorem~\ref{th:triagnulationtore} is proved :

\begin{theorem}[\cite{GL13}]
\label{th:3connected}
An essentially 3-connected toroidal map admits a crossing 
Schnyder wood.
\end{theorem}

Theorem~\ref{th:3connected} is stronger than
Theorem~\ref{th:triagnulationtore} for two reasons. First, it
considers essentially 3-connected toroidal maps and not only
triangulations, thus it proves Conjecture~\ref{conjecture2} for
$g=1$. Second, it shows the existence of crossing Schnyder woods.

However, what we have done in Section~\ref{sec:proof} for triangulation
can be generalized to essentially 3-connected toroidal maps. For that
purpose one has to work in the primal-dual completion. Proofs get more
technical and instead of walks in the primal now walks in the dual of the primal-dual completion have to
be considered. This is why we restrict ourselves to triangulations.

Even if we did not prove the existence of crossing Schnyder woods, 
Lemma~\ref{lem:2middlecycle} gives a bit of crossing in the following sense.
A 3-orientation obtained by Lemma~\ref{lem:2middlecycle}
has two middle cycles that are not weakly homologous. Thus in the
corresponding  Schnyder wood, these two cycles correspond to
two monochromatic cycles that intersect. We say that
the Schnyder wood obtained
by Theorem~\ref{th:triagnulationtore} is \emph{half-crossing}, i.e.,
 there exists a pair $i,j$ of different colors, such that there exist a monochromatic cycle
of color $i$ intersecting a monochromatic cycle of color $j$.

A half-crossing Schnyder wood is not necessarily
crossing. The 3-orientation of the toroidal triangulation of
Figure~\ref{fig:halfcrossing} is an example where two middle cycles
are not weakly homologous, so it corresponds to a half-crossing
 Schnyder wood. However, It is not crossing because the green and the
blue cycle do not intersect.

\begin{figure}[!h]
\center
\includegraphics[scale=0.5]{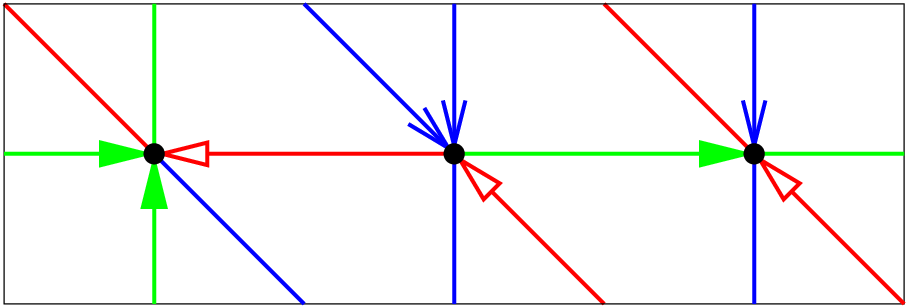} 
\caption{A not crossing but half-crossing Schnyder wood of a toroidal triangulation.}
\label{fig:halfcrossing}
\end{figure}

Consider a toroidal triangulation $G$ and a pair $\{B_1,B_2\}$ of cycles
that form a basis for the homology.  Figure~\ref{fig:trimap} shows how
to transform an orientation of $G$ into an orientation of
$\pdc{G}$. With this transformation a Schnyder wood of $G$ naturally
corresponds to a Schnyder orientation of $\pdc{G}$. This allows us to
not distinguish between a Schnyder wood or the corresponding Schnyder
orientation of $\pdc{G}$.  Recall from
Section~\ref{sec:transformations}, that the type of a Schnyder
orientation of $\pdc{G}$ in the basis $\{B_1,B_2\}$ is the pair
$(\gamma(B_1),\gamma(B_2))$.

\begin{lemma}
\label{lem:crossingtype}
A half-crossing Schnyder wood is of type $(0,0)$ (for the considered
basis).
\end{lemma}

\begin{proof}
  Consider a half-crossing Schnyder wood of $G$ and $C_1,C_2$ two
  crossing monochromatic cycles. We have
  $\gamma(C_1)=\gamma(C_2)=0$. The cycles $C_1,C_2$ are not
  contractible and not weakly-homologous. So by
  Lemma~\ref{lem:deltagammatri}, any non-contractible cycle $C$ of $G$
  satisfies $\gamma(C)=0$. Thus $\gamma(B_1)=\gamma(B_2)=0$.
\end{proof}

A consequence of Lemma~\ref{lem:crossingtype} is the following:

\begin{theorem}
\label{cor:halfcrossing}
  Let $G$ be a toroidal triangulation, given with a particular half-crossing
  Schnyder wood $D_0$, then the set $T(G,D_0)$ of all 
  Schnyder woods of ${G}$ that have the same type as $D_0$ contains
  all the half-crossing Schnyder woods of $G$.
\end{theorem}

Recall from Section~\ref{sec:lattice}, that the set $T(G,D_0)$ carries
the structure of a distributive lattice. This lattice contains all the 
half-crossing Schnyder woods. It shows the existence of a canonical 
lattice useful for bijection purpose, see~\cite{DGL15}.

Note that $T(G,D_0)$ may contain Schnyder woods that are not
half-crossing.  The  Schnyder wood of Figure~\ref{fig:gamma0}
is an example where $\gamma(C)=0$ for any non-contractible cycle
$C$. So it is of the same type as any half-crossing Schnyder wood but
it is not half-crossing.

Note also that in general there exist Schnyder woods not
in $T(G,D_0)$. The Schnyder wood of Figure~\ref{fig:gammanot0}
is an example where the horizontal cycle has $\gamma$ equal to $\pm
6$. Thus it cannot be of the same type as a half-crossing Schnyder
wood.

\begin{figure}[!h]
\center
\includegraphics[scale=0.5]{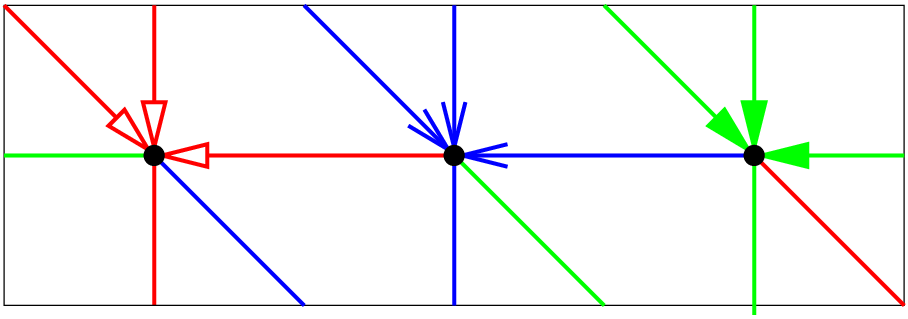} 
\caption{A  Schnyder wood of a toroidal triangulation where
  $\gamma(C)\neq 0$ for a non-contractible cycle $C$.}
\label{fig:gammanot0}
\end{figure}

\subsection{A lattice example}
\label{sec:example}

Figure~\ref{fig:lattice} illustrates the Hasse diagram of the set
$T(G,D_0)$ for the toroidal triangulation $G$ of
Figure~\ref{fig:halfcrossing}. Bold black edges are the edges of the
Hasse diagram $\mathcal{H}$. Each node of the diagram is a Schnyder
wood of $G$. Since we are considering a triangulation $\pdc{G}$ is not
represented in the figure. Indeed, all
the edges of $\pdc{G}$ incident to dual-vertices are outgoing in any
Schnyder orientation of $G$, thus these edges are rigid and do not
play a role for the structure of the lattice. In every Schnyder wood,
a face is dotted if its boundary is directed. In the case of
 the special face $f_0$ the dot is black. Otherwise, the  dot is magenta 
 if the boundary cycle is oriented \ccw and cyan otherwise. 
 An edge in the Hasse diagram from $D$ to $D'$
(with $D\leq D'$) corresponds to a face oriented \ccw in $D$ whose
edges are reversed to form a face oriented \cw in $D'$, i.e., a magenta
dot is replaced by a cyan dot. The outdegree of a node is its
number of magenta dots and its indegree is its number of cyan
dots. By Proposition~\ref{lem:necessary}, all the faces have a dot
at least once. The special face is not allowed to be flipped, it
is oriented \ccw in the maximal Schnyder wood and \cw in the minimal
Schnyder wood by Proposition~\ref{lem:maxtilde}.  By
Proposition~\ref{prop:maximal}, the maximal (resp. minimal) Schnyder
wood contains no other faces oriented \ccw (resp. \cww), indeed in
contains only cyan (resp. magenta) dots.  The words ``no'', ``half'',
``full'' correspond to Schnyder woods that are not
half-crossing, half-crossing (but not crossing), and crossing, respectively. By
Theorem~\ref{cor:halfcrossing}, the figure contains all the
half-crossing Schnyder woods of $G$. The minimal element is the
Schnyder wood of Figure~\ref{fig:gamma0}, and its neighbor is the
Schnyder wood of Figure~\ref{fig:halfcrossing}. 

The graph is very symmetric so the lattice does not depend on the
choice of special face. In the example the two crossing Schnyder woods
lie in the ``middle'' of the lattice.  These Schnyder woods are of
particular interests for graph drawing (see~\cite{GL13}) whereas the
minimal Schnyder wood (not crossing in this example) is important for
bijective encoding (see~\cite{DGL15}).

\begin{figure}[!h]
\center
\includegraphics[scale=0.25]{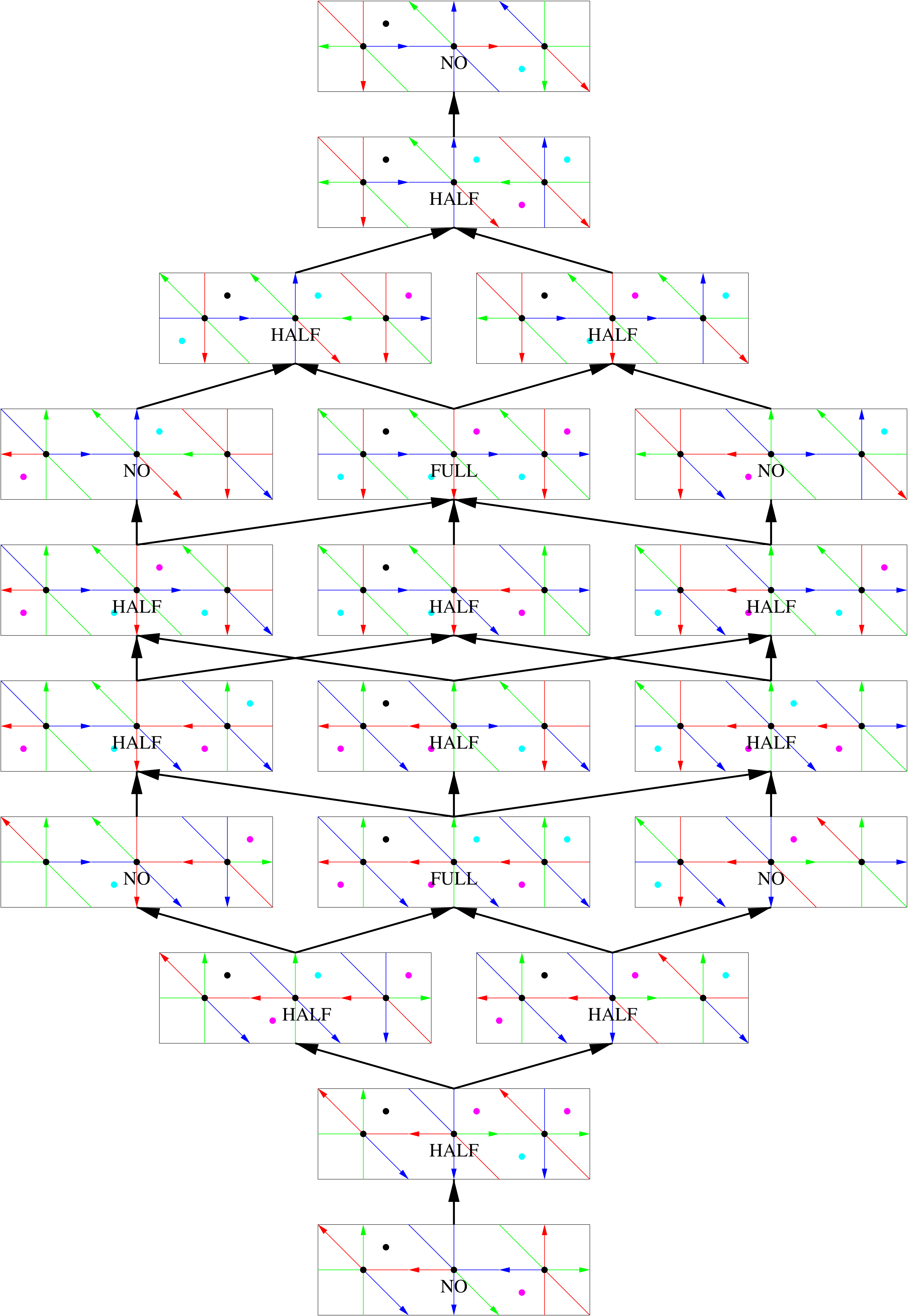} 
\caption{Example of the Hasse diagram of the distributive lattice of
  homologous orientations of a toroidal triangulation.}
\label{fig:lattice}
\end{figure}

The underlying toroidal triangulation of Figure~\ref{fig:lattice} has
only two Schnyder woods not depicted in
Figure~\ref{fig:lattice}.
One of them two Schnyder wood is shown in Figure~\ref{fig:gammanot0}
and the other one is a 180\textdegree rotation of Figure~\ref{fig:gammanot0}. 
Each of these Schnyder wood is  alone in its lattice of
homologous orientations. 
All their edges are rigid. They have no
0-homologous oriented subgraph. 

Theorem~\ref{th:transformations} says that one can take the Schnyder
wood of Figure~\ref{fig:gammanot0}, reverse three or six vertical cycle
(such cycles form an Eulerian-partitionable oriented subgraph) to obtain
another Schnyder wood. Indeed, reversing any three of these cycles leads
to one of the Schnyder wood of Figure~\ref{fig:lattice} (for example
reversing the three loops leads to the crossing Schnyder wood of the
bottom part). Note that ${3 \choose 6}=20$ and there are exactly twenty
Schnyder woods on Figure~\ref{fig:lattice}. Reversing six cycles leads
to the same picture pivoted by 180°.

\section{Conclusions}

In this paper we propose a generalization of Schnyder woods to higher
genus via angle labelings. We show that these objects behave nicely
with simple characterization theorems and strong structural
properties. Unfortunately, we are not able to prove that every essentially 3-connected map admits
a generalized Schnyder wood.

As mentioned earlier, planar Schnyder woods have applications in
various areas. In the toroidal case, they already lead to some results
concerning graph drawing~\cite{GL13} and optimal
encoding~\cite{DGL15}.  It would be interesting to see which other
applications can be generalized to higher genus. 

Note also that the distributive lattice structure of homologous
orientations of a given map (see Theorem~\ref{th:lattice}) is a very
general result that may be useful to study other objects (transversal
structures, $\frac{d}{d-2}$-orientations, etc.) associated to other
kinds of maps (4-connected triangulations, d-angulations, etc.).

\end{document}